%% file: k_cliques.tex
\documentclass[11pt]{article}
\usepackage{fullpage}
\usepackage{dsfont}
\usepackage{mathrsfs}
\usepackage{titling}

\def\conf{0}
\def\withcolors{0}
\def\stoc{0}

\input{preamble}

\input{drawings}

	\newcommand{\clk}{C_k}
\newcommand{\sclk}{c_k}
\newcommand{\onk}{\overline{C}_k}
\newcommand{\oa}{\overline{\va}}
\newcommand{\ta}{\widetilde{\va}}
\newcommand{\ha}{\widehat{\va}}

\renewcommand{\emph}[1]{\textup{\texttt{#1}}}

\newcommand{\va}{v}

\newcommand{\mA}{\mathcal{A}}
\newcommand{\mC}{\mathcal{C}}

\newcommand{\ap}{\alpha_{P}}
\newcommand{\dm}{d}

\newcommand{\hnk}{\widehat{C}_k}
\newcommand{\shnk}{\widehat{c}_k}

\newcommand{\oeps}{{\overline{\eps}}}
\newcommand{\odelta}{{\overline{\delta}}}

\newcommand{\thres}{2\sqrt{\om}}
\newcommand{\mrange}{[(1-\oeps)m,m]}
\newcommand{\krange}{[\clk/4,\clk]}

\newcommand{\kthres}{k\cdot(50\onk)^{1-1/k}/\oeps^{1/k}}
\newcommand{\sett}{\frac{10n\cdot \ln (2n/\gamma^{\hspace{1pt}2})}{(\oeps/k)^2 \cdot \thres}}
\newcommand{\sets}{\frac{700\cdot k\cdot n\cdot\ln(1/\odelta)}{\oeps^{\hspace{1pt}2+{1/k}}\cdot  \onk^{1/k}}}
\newcommand{\setq}{\frac{m(S)\cdot (\thres)^{k-2}}{(1-\oeps)^{3}\cdot (k-2)!\cdot \onk \cdot (s/n)}\cdot \frac{10\ln(1/\odelta)}{\oeps^{\hspace{1pt}2}}}
\newcommand{\mk}{4\om/(\oeps \onk)^{1/k}}
\newcommand{\setgamma}{\min\{1/(4\om^{k/2}),\odelta\} }

\newcommand{\qpopular}{\frac{d(u)\cdot (\thres)^{k-2}}{(k-2)!\cdot (\kthres)}\cdot \frac{15\ln(n/\odelta)}{\oeps^{\hspace{1pt}2}}}

\newcommand{\invokeA}{\mA\left(\oa,\eps,\delta, \overrightarrow{V}\right)}
\newcommand{\invokeAp}{\mA\left(\oa,\eps,\delta, \overrightarrow{V}\right)}
\newcommand{\ert}[1]{E_{rt}\left(#1\right)}
\newcommand{\Pip}{P_{is}}
\newcommand{\link}[1]{\hyperref[#1]{\color{black}\textsf{#1}}}

\newcommand{\popular}{sociable}
\newcommand{\Popular}{Sociable}
\newcommand{\unpopular}{shy}

\begin{document}

\begin{titlepage}

	\title{On Approximating the Number of $k$-cliques in Sublinear Time}

	\author{
		Talya Eden
		\thanks{Tel Aviv University, \textit{ talyaa01@gmail.com}. This research was partially supported by a grant from the Blavatnik fund. The author is grateful to the Azrieli Foundation for the award of an Azrieli Fellowship.}
		\and  Dana Ron
		\thanks{Tel Aviv University, \textit {danaron@tau.ac.il}.
		This research was partially supported by the Israel Science Foundation grant No. 671/13 and by a grant from the Blavatnik fund.}
	\and  C. Seshadhri
		\thanks{University of California, Santa Cruz, \textit{ sesh@ucsc.edu}}
	}

	\date{}
	\maketitle
	\begin{abstract}
		\input{abstract}
	\end{abstract}

\thispagestyle{empty}

\end{titlepage}
	
\setcounter{page}{1}
	\newpage

\input{intro-attempt}

\input{prel}	
\input{alg}

\ifnum\conf=0

\input{search}

\fi

\ifnum\conf=0	
	\bibliographystyle{alphaurl}
\else
	\bibliographystyle{alphaurl}
\fi
	\bibliography{k_cliques_bib}			
	
\def\accompanying{0}

\ifnum\accompanying=1
\ifnum\conf=1

\def\conf{0}
\setcounter{equation}{0}
\setcounter{section}{0}
\setcounter{subsection}{0}
\setcounter{theorem}{0}
\setcounter{figure}{0}
\setcounter{footnote}{0}

\newpage

\begin{titlepage}
	
	\title{On Approximating the Number of $k$-cliques in Sublinear Time
		\ifnum\conf=1
		\\ {\small (Extended Abstract)}
		\else
		\\ {\small (Full Version)}
		\fi
	}
	
	\author{
		Talya Eden
		\thanks{Tel Aviv University, \textit{ talyaa01@gmail.com}. This research was partially supported
			by a grant from the Blavatnik fund. The author is grateful to the Azrieli Foundation for the award of an Azrieli Fellowship.}
		\and  Dana Ron
		\thanks{Tel Aviv University, \textit {danaron@tau.ac.il}.
			This research was partially supported
			by the Israel Science Foundation grant No. 671/13 and by a grant from the Blavatnik fund.}
		\and  C. Seshadhri
		\thanks{University of California, Santa Cruz, \textit{ sesh@ucsc.edu}}
	}

	\date{}
	\maketitle
	\begin{abstract}
		\input{abstract}
	\end{abstract}
	
	\thispagestyle{empty}
	
\end{titlepage}

\setcounter{page}{1}

\input{intro-attempt}
\input{prel}
\input{alg}
\input{search}
\input{lb}

\fi

\fi 
\end{document}

%% file: preamble.tex
\usepackage{etex}

\usepackage[margin=2.54cm]{geometry}
\usepackage[dvipsnames,usenames]{color}
\usepackage{amsfonts,amsmath,amssymb,amsthm}

\usepackage[margin=10pt]{subfig}

\usepackage{paralist}
\usepackage{algorithmic,algorithm}
\usepackage{bm}
\usepackage{xspace}
\usepackage{xspace,prettyref}
\usepackage{color}
\usepackage{dsfont}

\usepackage[colorlinks,urlcolor=blue,citecolor=blue,linkcolor=blue]{hyperref}

\newcommand{\ignore}[1]{}

\newtheorem{theorem}{Theorem}

\newtheorem{lemma}[theorem]{Lemma}
\newtheorem{claim}[theorem]{Claim}

\newtheorem{corollary}[theorem]{Corollary}

\newtheorem{definition}[theorem]{Definition}

\theoremstyle{definition}

\newcommand{\EX}{{\rm Exp}}

\renewcommand{\th}{^\textrm{th}}

\def\eqdef{~\triangleq~}
\def\eps{\epsilon}

\ifnum\stoc=0
\newenvironment{proofof}[1]{\smallskip\noindent{\bf Proof of #1:}}%
        {\hspace*{\fill}$\Box$\par}
\else
\newenvironment{proofof}[1]{\smallskip\noindent{\scshape Proof of #1:}}%
        {\hspace*{\fill}$\qed$\par}
\fi

\newcommand{\om}{\overline{m}}

\newcommand{\Sec}[1]{\hyperref[sec:#1]{\S\ref*{sec:#1}}} 
\newcommand{\Eqn}[1]{\hyperref[eq:#1]{(\ref*{eq:#1})}} 
\newcommand{\Fig}[1]{\hyperref[fig:#1]{Fig.\,\ref*{fig:#1}}} 
\newcommand{\Tab}[1]{\hyperref[tab:#1]{Tab.\,\ref*{tab:#1}}} 
\newcommand{\Thm}[1]{\hyperref[thm:#1]{Theorem\,\ref*{thm:#1}}} 
\newcommand{\Lem}[1]{\hyperref[lem:#1]{Lemma\,\ref*{lem:#1}}} 
\newcommand{\Prop}[1]{\hyperref[prop:#1]{Prop.~\ref*{prop:#1}}} 
\newcommand{\Cor}[1]{\hyperref[cor:#1]{Corollary~\ref*{cor:#1}}} 
\newcommand{\Def}[1]{\hyperref[def:#1]{Definition~\ref*{def:#1}}} 
\newcommand{\Alg}[1]{\hyperref[alg:#1]{Alg.~\ref*{alg:#1}}} 
\newcommand{\Ex}[1]{\hyperref[ex:#1]{Ex.~\ref*{ex:#1}}} 
\newcommand{\Clm}[1]{\hyperref[clm:#1]{Claim~\ref*{clm:#1}}} 
\newcommand{\Assum}[1]{\hyperref[assump:#1]{Assumption~\ref*{assump:#1}}} 

\def\poly{{\rm poly}}

\newcommand{\mnote}[1]{\marginpar{\tiny\bf
		\begin{minipage}[t]{0.5in}
			\raggedright #1
		\end{minipage}}}

\newcommand\numberthis{\addtocounter{equation}{1}\tag{\theequation}}

\renewcommand{\Pr}{\mathrm{Pr}}

\newcommand{\parmath}[1]{
\ifnum\stoc=0
\boldmath{#1}
\else
#1
\fi
}

\newcommand{\confeqn}[1]{
	\ifnum\stoc=0
	\[#1\]
	\else
	$#1$
	\fi
}

\newcommand{\confFigure}[2]
{
	\ifnum\stoc=1
	\begin{figure}[htb!]
		\fbox{
			\begin{minipage}{0.465\textwidth}
				\begin{raggedright}
					#1					
				\end{raggedright}
			\end{minipage}
		}
		\caption{#2}
	\end{figure}
	\else
	\begin{figure}[htb!]
		\fbox{
			\begin{minipage}{0.9\textwidth}
				\begin{raggedright}
					#1					
				\end{raggedright}
			\end{minipage}
		}
		\caption{#2}
	\end{figure}
	
	\fi
}

\def\nofigures{0}

\ifnum\nofigures=0
\usepackage{pstricks,tikz}
\usetikzlibrary{decorations.markings}
\usetikzlibrary{patterns}
\pgfdeclarelayer{background}
\pgfsetlayers{background,main}
\fi

\ifnum\withcolors=1

\newcommand{\tcom}[1]{{{#1}}}
\newcommand{\dcom}[1]{{{#1}}}
\newcommand{\scom}[1]{{{#1}}}
\newcommand{\dddchange}[1]{{\color{cyan}{#1}}}
\newcommand{\talya}[1]{{{#1}}}
\newcommand{\tchange}[1]{{\color{red}{#1}}}
\newcommand{\schange}[1]{{\color{blue}{#1}}}
\newcommand{\Sesh}[1]{{\color{red} Sesh says: #1}}

\newcommand{\Talya}[1]{{#1}}
\newcommand{\strikeout}[1]{{\color{gray}  #1}}

\else
\newcommand{\scom}[1]{{{#1}}}

\newcommand{\schange}[1]{{{#1}}}

\newcommand{\tchange}[1]{{{#1}}}
\newcommand{\dcom}[1]{{{#1}}}
\newcommand{\dddchange}[1]{{{#1}}}

\newcommand{\tcom}[1]{{{#1}}}
\renewcommand{\mnote}[1]{}
\newcommand{\talya}[1]{{{#1}}}
\newcommand{\Sesh}[1]{{#1}}

\newcommand{\Talya}[1]{{#1}}
\newcommand{\strikeout}[1]{{}}

\fi

%% file: drawings.tex
\usepackage{graphicx,amsfonts,amsmath,amssymb,epsfig,color,multicol,pstricks,tikz, pgf}
\usetikzlibrary{decorations.markings}
\usetikzlibrary{patterns}
\usetikzlibrary{calc, intersections}
\usepackage{pgf,tikz}
\usetikzlibrary{calc}
\usepackage{tkz-euclide}
\usepackage{pgfplots}
\usetkzobj{all}

\pgfdeclarelayer{background}
\pgfsetlayers{background,main}

\tikzset{dotted pattern/.style args={#1}{
		postaction=decorate,
		decoration={
			markings,
			mark=between positions 0.25 and 0.75 step 0.25 with {
				\fill[radius=#1] (0,0) circle;
			}
		}
	},
	dotted pattern/.default={1pt},
}

\tikzstyle{client}=[draw, circle, minimum size=1.5ex, inner sep=0,fill=black]
	
\newcommand{\drawBasic}[2]{
		\foreach \i in {1,2,...,3} {
			\node[client] (t\i) at (\i,3) {} ;
			\node[client] (b\i) at (\i,-3) {} ;
		}
		\foreach \i in {1,2,...,3} {
			\node[client] (l\i) at (-4,\i-2) {} ;
			\node[client] (r\i) at (-2,\i-2) {} ;
		}
		
		\node[] (label) at (-3,-2.1) {\large{#1}} ;
		\node[text width=3cm,align=center] (label) at (-3,-2.9) {{#2}} ;
		\node[label={[xshift=-.3cm, yshift=0.1cm]$A_1$}] (label) at (1,3) {};
		\node[label={[xshift=-.3cm, yshift=-.3cm]$A_2$}] (label) at (1,-3.5) {};

		\draw [thick, decoration={ brace, raise=0.5cm,amplitude=10pt}, decorate] (3,-3) -- (1,-3) node [midway, below, yshift=-0.8cm, align=center] {$\frac{\sqrt m}{k-2}$} ;
				
		\foreach \i in {1,2,3} {
			\foreach \j in {1,2,3} {
				\foreach \x in {l,r,t,b} {
					\foreach \y in {l,r,t,b} {
						\ifthenelse {\equal{\x}{\y} }	{} {
							\draw[ultra thin, gray] (\x\i) -- (\y\j);								
						};
					}
				}
			}
		}
		\path[dotted pattern=0.5pt] (t2) -- (t3);
		\path[dotted pattern=0.5pt] (b2) -- (b3);
		\path[dotted pattern=0.5pt] (l1) -- (l2);
		\path[dotted pattern=0.5pt] (r1) -- (r2);
			
}

%% file: abstract.tex
We study the problem of approximating the number of $k$-cliques in a graph when given query access to the graph.
 We consider the  standard query model for general graphs via
(1) degree queries, (2) neighbor queries and (3) pair queries.
Let $n$ denote the number of vertices in the graph, $m$ the number of edges, and $\clk$ the number of $k$-cliques.
We design an algorithm that outputs a $(1+\eps)$-approximation (with high probability) for $\clk$, whose expected query complexity and running time are
$O\left(\frac{n}{\clk^{1/k}}+\frac{m^{k/2}}{\clk} \right)\poly(\log n, 1/\eps,k)$.
Hence, the complexity of the algorithm is sublinear in the size of the graph for $\clk = \omega(m^{k/2-1})$.
Furthermore, 
the query complexity of our algorithm is essentially optimal
(up to the dependence on $\log n$, $1/\eps$ and $k$).
%
%

The previous results in this vein are by Feige (SICOMP 06) and by Goldreich \dcom{and} Ron (RSA 08)
for edge counting ($k=2$) and by \dcom{Eden et al.} 
(FOCS 2015)
for triangle counting ($k=3$). Our result matches the complexities of
these results.

\scom{
	\Talya {The previous result by Eden et al. hinges on a certain amortization technique that works 
for triangle counting,
and does not generalize 
\dddchange{to all $k$.}
We 
\dcom{obtain} a general algorithm that works for any $k\geq 3$ by designing a procedure that} samples each $k$-clique
incident to a given set $S$ of vertices with approximately equal probability.
The primary difficulty is in finding cliques incident to purely high-degree vertices,
since random sampling within neighbors has a low success probability.
This is achieved by an algorithm that samples uniform random
high degree vertices and
a careful tradeoff between estimating cliques incident purely to high-degree vertices and those
that include a low-degree vertex.}

%% file: intro-attempt.tex
\section{Introduction}\label{sec:intro}

Counting the number of $k$-cliques in a graph is a classic problem in theoretical
computer science, 
and the special case of $k=3$ (triangle counting)
is itself an important problem. In practice, clique counting has received
much attention due to its significance for analyzing real-world graphs~\cite{HoLe70,Co88,portes2000social,EcMo02,milo2002network,Burt04,becchetti2008efficient,foucault2010friend,BerryHLP11,SeKoPi11,JRT12,ELS13,Ts15,FFF15}.
From a theoretical standpoint, the best exact
algorithms use matrix multiplication~\cite{NP85,EG04}. Better bounds for sparse graphs can
be obtained by combinatorial methods~\cite{ChNi85,V09}.

There is a long line of  algorithms for approximately counting the number of cliques (especially
triangles)  in various computational models, including distributed and streaming settings ~\cite{ChNi85,ScWa05,ScWa05-2,jowhari2005new,tsourakakis2008fast,tsourakakis2009doulion,avron2010counting,kolountzakis2012efficient,chu2011triangle,SuVa11,tsourakakis2011triangle,AhGuMc12,KaMeSaSu12,SePiKo13,TaPaTi13}.
All these algorithms begin by reading their entire input graph, \talya{and hence must run in at least linear time}.
Recently, Eden et al.~\cite{ELRS} gave the first sublinear-time algorithm for
triangle counting. The query model
used is the standard model for sublinear algorithms
on general graphs \schange{(refer to Chapter 10 of Goldreich's book~\cite{G17-book})}. We assume
that the vertex set is $V = [n]$. Algorithms can make the following
queries.
(1) Degree queries: given $v \in V$, get the degree $d(v)$.
(2) Neighbor queries, given $v \in V$ and $i \leq d(v)$
get the $i\th$ neighbor of $v$.
(3) Pair queries: given vertices $u,v$, determine if $(u,v)$
is an edge.

We show that  there is a sublinear-time algorithm for approximating
the number of $k$-cliques in this model, for any given $k$, subsuming the
result of~\cite{ELRS} for $k=3$.


\subsection{Results}
Let $G=(V,E)$ be a graph over $n$ vertices and $m$ edges, where we view edges as ordered pairs, so that $m$ equals the sum of the degrees of all vertices.
Let $\clk$ denote that number of cliques of size $k$ in $G$.

 Our main theorem follows.

\begin{theorem} \label{thm:main}
There exists an algorithm
that, given $n,\;k$, an approximation parameter $0<\eps<1$, and query access to a graph $G$,
outputs an estimate $\hnk$, such that with high constant probability (over the randomness of the algorithm), $$(1-\eps)\cdot \clk\leq \hnk\leq (1+\eps)\cdot \clk.$$
The expected
query complexity of the algorithm is
	\[O\left(\frac{n}{\clk^{1/k}}+\min\left\{\frac{m^{k/2}}{ \clk}\dcom{,m}\right\}\right)\cdot\poly(\log n, 1/\eps, k)  \;,\]
and the expected running
time is
$O\Big(\frac{n}{\clk^{1/k}}+\frac{m^{k/2}}{ \clk}\Big)\cdot\poly(\log n, 1/\eps, k)$.
\end{theorem}
%
\tchange{Below we state a nearly matching lower bound. The lower bound was first given in an earlier version of this paper~\cite{ERS_cliques}, and a simplified proof is given in \cite{ER_CC}}.
\begin{theorem}[\cite{ERS_cliques, ER_CC}] \label{thm:main-lb}
The query complexity of any multiplicative-approximation
algorithm for $\clk$ is
\[\Omega\left(\frac{n}{\clk^{1/k}}+\min\left\{\frac{m^{k/2}}{\clk \cdot (c \cdot k)^k}\dcom{,m}\right\}\right) \;,\]
for a (small) constant $c$.
\end{theorem}

\subsection{Main ideas and techniques}\label{subsec:intro-alg}
In all that follows, a random vertex refers to a vertex selected uniformly at random.
\dcom{Our starting point for} approximating the number of $k$-cliques
\dcom{is similar} to that of Eden et al.~\cite{ELRS} (henceforth ELRS)
for 
\dcom{approximating} the number of triangles (i.e., $k=3$). 
\dcom{However, we diverge} rather quickly,
as the heart of our algorithm, a process for sampling $k$-cliques, is  conceptually and technically
different.
%
%
For the sake of simplicity, assume $\eps$ is a constant. Also, it is convenient
to assume that constant factor estimates of $\clk$ and  $m$ are known.\footnote{The assumption on $m$ can be removed by \dddchange{applying~\cite{feige2006sums},} 
and the assumption on $\clk$ can be removed by a geometric search (for details see
\ifnum\conf = 0
Subsection~\ref{subsec:search}).
\else
  Subsection 3.7 in the full version).
\fi
}

\paragraph{The starting point that is shared with ELRS: vertex sampling and clique assignment.}
Our algorithm starts by uniformly and independently selecting a (multi-)set $S$ of vertices of size roughly $n/\clk^{1/k}$.
Denoting the number of $k$-cliques incident to a vertex  $u$ by $c_k(u)$, a natural
estimate for $C_k$ is $\frac{n}{k|S|}\sum_{u \in S} c_k(u)$. Unfortunately, for a  random $u$,
$c_k(u)$ can have extremely large variance. ELRS
reduce the variance by considering triangles only from endpoints $u$
where $c_3(u)$ is not too large. In our setting, we formalize this  by defining ``\popular" vertices.
A vertex is  \popular\ if it participates in a number of $k$-cliques that is above a certain threshold $\tau$ or if its degree is above another threshold $\tau'$. 
A $k$-clique with vertices $\{u_1,\dots,u_k\}$ is {\em assigned} to the smallest degree vertex $u_j$
that is not \popular. We set the parameters $\tau, \tau'$ to ensure that the number of $k$-cliques that are not assigned to any vertex is at most $\eps \cdot C_k$.
Let $\alpha(u)$ be the number of $k$-cliques assigned to $u$.
We can prove that if $|S|$ is roughly $n/C_k^{1/k}$, then
$\frac{n}{|S|} \sum_{u \in S} \alpha(u)$ is a $(1+\eps)$-approximation
of $C_k$ with high probability. The problem of approximating $C_k$ is now reduced to
approximating $\alpha(S) = \sum_{u \in S} \alpha(u)$. \Talya{From this point on the ELRS approach does not generalize (as we explain towards the end of this section). 
\dcom{We continue by describing our approach.}
}

%
%
%
%
%
%
%

\paragraph{
Approximating \parmath{$\alpha(S)$}: Sampling \parmath{$k$}-cliques incident to \parmath{$S$} almost uniformly.
}
Let $\mA(S)$ denote the (multi-)set of $k$-cliques
assigned to vertices in $S$.
We estimate $\alpha(S) = |\mA(S)|$ by sampling
$k$-cliques incident to $S$, and checking whether they
belong to $\mA(S)$. The key is to sample
each $k$-clique  incident to $S$ (and in particular in $\mA(S)$) with (approximately) {\em the same probability\/}.
Specifically, as explained next, we will do so with probability proportional to $\frac{1}{m(S)\cdot m^{k/2\dcom{-1}}}$, where $m(S) = |E(S)|$ \dcom{and $E(S)$ denotes the (multi-)set of (ordered) edges incident to $S$}.

Each $k$-clique containing a vertex $u\in S$ is associated with one of the edges incident to $u$ in the clique. For reasons that will become clear shortly, the other endpoint of this edge is selected to be the lowest degree vertex among the vertices in the clique (excluding $u$, and breaking ties arbitrarily but consistently).
The procedure for sampling $k$-cliques starts by sampling an edge $(u,v)$ uniformly in $E(S)$.
It then attempts to extend this edge to a $k$-clique, where the other $k-2$ vertices have degree higher  than $v$.
This is done in one of two different ways, depending on the degree of $v$,
where in either way, the algorithm may fail to output any $k$-clique.
\talya{For the sake of simplicity, } in what follows \Talya{we refer to a vertex as} a {\em low-degree vertex} if its degree is at most $\sqrt{m}$, and as a {\em high-degree vertex} otherwise.
\talya{(The technical part of the paper uses a slightly different definition.)}

\medskip{\sf The \Talya{(easy)} ``low case'':}~ If $v$ is a low-degree vertex, then we
sample $k-2$ random neighbors of $v$ and check whether we obtained a $k$-clique in which $v$ is the lowest degree vertex other than $u$. 
In order to ensure that all $k$-cliques incident to $S$ 
are output with the same probability, we apply
rejection sampling and keep each sampled neighbor $w$ of $v$ with probability $d(v)/\sqrt m$ (and conditioned on
$d(v) \leq d(w)$).
Hence, each such clique (i.e., in which the lowest degree vertex other than $u\in S$ is a low-degree vertex)
is output
with equal probability  $\frac{(k-2)!}{m(S)\cdot \sqrt{m}^{k-2}}$.  (The $(k-2)!$ factor is due to the fact that we may obtain the $k-2$ vertices in the clique other than $u$ and $v$ in any order.)

 \medskip{\sf The \Talya{(hard)} ``high case'':}~The challenging case is
when $v$ is a high-degree vertex.
\scom{Rejection sampling, as in the low case, is too expensive now.}
However, observe that we are interested only in sampling neighbors of $v$ with degree higher than $v$,
and that the number of vertices with high degree is at most $\sqrt{m}$.  Therefore, if we
had a way to efficiently sample each high-degree vertex with probability (approximately) $1/\sqrt{m}$,
we would obtain the same probability over cliques 
incident to $S$ as in the low case.
We next \dcom{explain} how this can be done.

Consider selecting a  random multi-set $T$ of roughly $t= \frac{n}{\log n/\sqrt{m}}$ vertices.
The setting of $t$ is such that with high probability, for every high-degree vertex $w$, the number of neighbors that $w$ has in $T$ is close to its expected value, that is, $d(w)\cdot \frac{t}{n}$. This implies that if we select an edge $(x,y)$ uniformly at random in $E(T)$ (the (ordered) edges incident to $T$, whose number is $m(T)$), then the probability that $y = w$ for a fixed high-degree vertex $w$, is approximately
$\frac{d(w)\cdot (t/n)}{m(T)}$.

Assume that $m(T)$ is not much larger than its expected value, $m\cdot\frac{t}{n}$ (which can be ensured with high probability). Let $p(w) = \frac{m(T)}{d(w)\cdot (t/n)\cdot \sqrt{m}}$, so that under this assumption on $m(T)$, and the fact that $w$ is a high-degree vertex, $p\in (0,1]$.
If we now keep $w$ with probability $p(w)$, then we have a subroutine that samples each high-degree vertex with approximately equal probability $1/\sqrt{m}$.
This in turn implies that we can select any fixed subset of $k-2$ high-degree vertices with probability very close to $\frac{(k-2)!}{\sqrt{m}^{k-2}}$. We then check whether whether the chosen vertices form a clique together with $u$ and $v$ (and that the clique is associated with $(u,v)$).

\medskip
 We now have a procedure that outputs each clique incident to $S$ 
 with (roughly) the same probability, $\frac{(k-2)!}{m(S)\cdot \sqrt{m}^{k-2}}$.
 In the next paragraph, we discuss a procedure for deciding whether a $k$-clique
 belongs to $\mA(S)$. Given this decision procedure,
 we can estimate $|\mA(S)|$ by performing
 $\frac{m(S)\cdot (2\sqrt{m})^{k-2}}{(k-2)!\cdot |\mA(S)| }$
 $= O\Big(\frac{m^{k/2}}{\clk}\Big)$ calls
 to the  $k$-clique sampling procedure.  
 (\Talya{We assume} that $|\mA(S)|$ is close to its expected value, $\clk\cdot \frac{s}{n}$
 and that $m(S)$ is not much larger than its expected value, $m\cdot \frac{s}{n}$).

\paragraph{Deciding whether a \parmath{$k$}-clique is in \parmath{$\mA(S)$}.}
Deciding
whether a $k$-clique 
incident to a vertex $u\in S$ should be assigned to $u$, requires to determine which of the clique vertices are \popular.
Recall that a vertex is considered \popular\ if the number of $k$-cliques it  participates in is more than $\tau$ or if its degree is above $\tau'$ (for appropriate settings of $\tau$ and $\tau'$). The second condition can be easily verified by a single degree query.
As for the first condition, given a vertex $u$, we verify whether the number of $k$-cliques that it participates in is more than $\tau$ by running the $k$-clique sampling procedure with $S=\{u\}$ \Talya{for} a sufficient number of times
(roughly $\frac{d(u)\cdot (2\sqrt{m})^{k-2}}{(k-2)!\cdot\tau} \leq \frac{\tau'\cdot (2\sqrt{m})^{k-2}}{(k-2)!\cdot\tau}$ where $\frac{\tau'}{\tau} =O\left(\frac{m}{\clk}\right)$).
The procedure may err on ``almost \popular" vertices, but the
analysis can be modified to deal with this.
While the procedure \tchange{might have} high query complexity, it is only invoked when the clique-sampling procedure returns a clique.
The frequency of the latter can be bounded appropriately
to get the final query complexity.

\scom{\paragraph{Why the ELRS algorithm does not generalize.} 
	The success of ELRS hinges on the following bound: $\sum_{(u,v) \in E} \min(d(u), d(v)) = O(m^{3/2})$,
	discovered in~\cite{ChNi85} in the context of triangle enumeration.
	The short answer to why the ELRS algorithm does not generalize is that the above bound does not have
	analogues for $k > 3$. Indeed, this is why the simple algorithm of~\cite{ChNi85} for triangle enumeration does not work
	for cliques of larger size.\Talya{\footnote{We note that~\cite{ChNi85} also present a general algorithm for listing $k$-cliques for any $k$, but this algorithm and its analysis are more involved.}}
	
	Let us revisit the triangle estimator of ELRS. \dcom{Recall that in this context, $\alpha(u)$ denotes the
number of triangles assigned to a vertex $u$, and that the}
aim is to estimate the average value of $\alpha(u)$ for $u \in S$.
	ELRS first ``transfer'' the assignment of triangles from vertices to edges.
	Letting $\alpha(u,v)$ denote the number of triangles assigned to the edge $(u,v)$
	we have that $\alpha(S) = \sum_{(u,v) \in E(S)} \alpha(u,v)$.
	For a random edge $(u,v)$,
a triangle can be detected by sampling a random neighbor of \Talya{the} lower degree
	vertex among $u$ and $v$ and performing  a pair query with the other endpoint.
Since the probability of finding a triangle decreases as $\min(d(u),d(v))$ \Talya{increases},
ELRS select $\left\lceil\frac{\min(d(u),d(v))}{\sqrt{m}}\right\rceil$ random neighbors.
Hence, the expected number of queries performed per edge is
$\frac{1}{m}\sum_{(u,v)\in E}\left\lceil\frac{\min(d(u),d(v))}{\sqrt{m}}\right\rceil$, which by
the aforementioned bound from~\cite{ChNi85}, is $O(1)$. ELRS prove that in order to estimate $\alpha(S)$,
it suffices to sample $O(m^{3/2}/C_3)$ random edges from $E(S)$, and therefore they get the desired bound.

	\sloppy
	A generalization of ELRS to $k$-cliques would require the following bound:
	$\sum_{(u,v) \in E} \min(d(u), d(v))^{k-2} = O(m^{k/2})$. \textit{This bound is false for $k>4$.}
    \schange{\dddchange{To exemplify this,} consider a graph \dddchange{over the vertex set $\{1,\dots,n\}$,} with the following edges. There is an edge $(1,2)$, and both vertices $1$ and $2$ have an edge
    to all other vertices. The left-hand-side of the bound is $\Theta(n^{k-2})$, while the right-hand-side is $O(n^{k/2})$.
	Thus, the ELRS analysis depends on a seeming singularity for $k=3, 4$, and 
\dcom{does not} generalize 
\dddchange{to all $k$.}
}
	
}

\subsection{Related work} \label{sec:related}

A significant portion of the work on clique counting focuses on triangle counting.
Because our focus is on general $k$, we avoid a detailed discussion of results for triangle
counting. We point the interested reader to~\cite{ELRS}.

Ne\v{s}et\v{r}il and Poljak give the first non-trivial algorithm for $k$-clique counting
by reducing to matrix multiplication~\cite{NP85}. Specifically, their algorithm runs in
\dcom{time} $O(n^{\omega\lfloor k/3\rfloor + k(\mathrm{mod} 3)})$, where $\omega$ is the \Talya{matrix multiplication exponent}.
Eisenbrand and Grandoni refine this bound for certain values of $k$ by careful reductions
to rectangular matrix multiplication~\cite{EG04}. They also give better dependencies on $m$
for sparse graphs. The general dependence of the form $n^{\omega k/3}$ is believed to be
optimal. Recent work by Abboud et al. 
builds on this conjecture to prove hardness for various parsing algorithms~\cite{ABW15}.
More relevant to our work, Chiba and Nishizeki give an algorithm for $k$-clique enumeration,
based on the arboricity of the graph, from which the $O(n+m^{k/2})$ bound for general graphs follows immediately.
\schange{The use of degree/degeneracy orientations have appeared
in recent practical works on clique counting~\cite{FFF15,JS17}.
In 
\dddchange{the current work}
we design various primitives to sample random $k$-cliques, by either
extending smaller cliques or by sampling high degree vertices.
The idea of extending smaller cliques to large ones using degree orientations
is an important feature of previous practical approaches~\cite{FFF15,JS17}. It would
be of interest to see if the new techniques given by our result could be used
for practical algorithms.}
%
%

In the context of sublinear algorithms, our work follows a line of results on sublinear
estimation of subgraph counts. Our analysis builds on several techniques developed in these results.
The starting point is the average degree estimation results of
Feige~\cite{feige2006sums} and Goldreich and Ron~\cite{GR08}.
Gonen et al. 
generalize these techniques to estimate
the count of $k$-stars in sublinear time~\cite{GRS11}.
Eden et al. \cite{ERS16} further extended and simplified
all these results, and show connections between this
problem and the graph degeneracy. They also build on the basic ELRS framework.
Eden and Rosenbaum \cite{EdenRosenbaum} provide an algorithm for sampling edges almost uniformly, and our clique sampler uses
some of their ideas to sample high-degree vertices. \schange{Dasgupta et al. \cite{DKS}
and Chierichetti et al.~\cite{ChDa+16} consider sublinear algorithms (for average degree and related problems)
in a weaker model where uniform random vertices are not allowed. In practical settings, we can only
``crawl" a graph, which translates to performing random walks. Their results typically require
some assumption about the mixing time of the input graph $G$. Again, we believe this is an interesting direction
for future work, to consider weaker query models but stronger assumptions on graph structure.}

%

Other work on sublinear algorithms for estimating graph parameters \dddchange{(}\tchange{in the standard query model}\dddchange{)} include results on the minimum weight spanning tree \cite{DBLP:journals/siamcomp/ChazelleRT05, DBLP:journals/siamcomp/CzumajS09, DBLP:journals/siamcomp/CzumajEFMNRS05}, maximum matching \cite{nguyen2008constant, yoshida2009improved} and minimum vertex cover \cite{DBLP:journals/tcs/ParnasR07,nguyen2008constant,DBLP:journals/talg/MarkoR09, yoshida2009improved, hassidim2009local, onak2012near}.

\Talya{	
\ifnum\conf=1
	\subsection{Organization \dcom{of this extended abstract}}
	We start with some preliminaries in Section~\ref{sec:prel}. In Section~\ref{sec:ub} we provide a description of the main algorithm, a proof sketch of the main theorem and a full description of one of the main procedures that the algorithm uses. The rest of the procedures as well as 
the lower bound are provided in the
accompanying full version of the paper.
\fi
}

%% file: prel.tex
	\section{Preliminaries}\label{sec:prel}

	We consider simple undirected graphs over a set $V$ of $n$ vertices. It is convenient to think of the graph edges as ordered pairs, so that every edge is considered from both endpoints. We say that the ordered edge $(u,v)$  \textsf{originates} from the vertex $u$. We denote the set of all ordered edges by $E$ and
let $m \eqdef |E|$. We use the following notations.
\smallskip

\begin{asparaitem}
    \item $d(u)$: the degree of a vertex $u$ (the number of edges originating from $u$).
    Note that $\sum_{u \in V} d(u) = m$.
    \item $E(S)$, $m(S)$: $E(S) \eqdef \{(u,v)\mid u \in S\}$  and $m(S)\eqdef|E(S)| = \sum_{u \in S}d(u)$.
    \item $d_S(u)$: for any vertex $u$ and set of vertices $S$, $d_S(u)$ is the number
    of neighbors of $u$ in $S$.
    \item $\clk$, $\sclk(u)$: $\clk$ is the number of $k$-cliques in the given graph.
    For $u\in V$, $\sclk(u)$ is the number of $k$-cliques that $u$ participates in.
    Note that $\clk = \frac{1}{k}\cdot \sum_{u\in V}\sclk(u)$.
\end{asparaitem}

\smallskip
	
	We use $\prec$ to denote a total order over the graph vertices such that for every two vertices $u$ and $v$,
if $d(u) < d(v)$, then $u \prec v$, and if $d(u)=d(v)$, then the order between $u$ and $v$ is determined in an arbitrary but fixed manner (e.g., by vertex id).
	
Let $[r]\eqdef \{1,\ldots,r\}$ and let $(1\pm \alpha)^t \cdot x$ denote the interval $\left[ (1-\alpha)^t \cdot x, (1+\alpha)^t \cdot x\right]$.
	
	We make use of the following version of Chernoff's inequality \cite{chernoff}. Let $\chi_i$ for $i=1, \ldots, m$ be  random variables taking values in $[0,B]$, such that for every $i$, $\EX[\chi_i]=p$. Then
\ifnum\stoc=0
	\[ \Pr\left[\frac{1}{m}\sum\limits_{i=1}^m \chi_i > (1+ \gamma)\mu \right] < \exp\left( -\frac{\gamma^2 \mu m}{3B}\right) \;,\]
	and 
	\[ \Pr\left[\frac{1}{m}\sum\limits_{i=1}^m \chi_i < (1- \gamma)\mu \right] < \exp\left( -\frac{\gamma^2 \mu m}{2B}\right)\;.\]
\else
	\small
	\begin{align*}
	 & \Pr\left[\frac{1}{m}\sum\limits_{i=1}^m \chi_i > (1+ \gamma)p \right] < \exp\left( -\frac{\gamma^2 p m}{3B}\right) \; \mbox{ and }\;
	\\
& \Pr\left[\frac{1}{m}\sum\limits_{i=1}^m \chi_i < (1- \gamma)p \right] < \exp\left( -\frac{\gamma^2 p m}{2B}\right)\;.
\end{align*}
\normalsize
\fi

The proof of the following claim is similar to the proof of~\cite{ChNi85} for their exact clique enumeration
\ifnum\conf=1
algorithm.
\else
algorithm (and we include it here for the sake of completeness).
\fi

	\begin{claim}\label{clm:u.b. nk} For every graph $G$ with $m$ (ordered) edges and $\clk$ $k$-cliques,
\confeqn{\clk \leq m \cdot \binom{\sqrt m}{k-2} \;.}
	\end{claim}
\ifnum\conf=0
	\begin{proof}
Let $D$ be the DAG obtained by orienting edges in $G$ according to $\prec$. Let $d^+(v)$
    be the out-degree of vertex $v$ in $D$. Observe that $\max_v \{d^+(v)\} \leq \sqrt{m}$.
    (All $d^+(v)$ out-neighbors of $v$ have degree at least $d(v) \geq d^+(v)$. Thus, $d^+(v) \leq \sqrt{m}$.)
    The number of $k$-cliques where $v$ is the lowest vertex
    according to $\prec$ is at most ${d^+(v) \choose k-1}$. Thus,
    $\clk \leq \sum_v {d^+(v) \choose k-1} \leq {\sqrt{m} \choose k-2} \sum_v d^+(v) = m \binom{\sqrt{m}}{k-2}$.
	\end{proof}

\fi	

%% file: alg.tex
\section{The algorithm}\label{sec:ub}

The main algorithm for approximating the number of $k$-cliques is presented in Subsection~\ref{subsec:main-alg}
and named \textsf{Approximate-cliques}. It takes the following parameters.

\medskip
\begin{asparaitem}
    \item $\om$: This is assumed to be a fairly precise estimate of the number of (ordered) edges $m$, and can be obtained using~\cite{GR08} (in expected time $O\left(\frac{n}{\sqrt{m}}\right)\cdot\poly(\log n,1/\eps)$).
    \item $\onk$: This is assumed to be a constant-factor estimate of $\clk$, which is obtained by geometric search (as shown     
    \ifnum\conf=0
    in Subsection~\ref{subsec:search}).
    \else
    in Subsection 3.7 of the  full version).
    \fi
    \item $\eps$: The main approximation parameter. We  set $\oeps = \eps/5$.
    \item $\delta$: The failure parameter. We set $\odelta = \delta/4$.
\end{asparaitem}

%
\subsection{\Popular\ vertices and the assignment of cliques to vertices}\label{subsec:alg-defs}	

The notion of {\em \popular} vertices, defined next,
is critical in reducing the variance of the output of our algorithm.

	\begin{definition}[\Popular\ and \unpopular\ vertices]\label{def:popular-v}
		\sloppy
		We say that a vertex $u$ is \textsf{\popular} if $\sclk(u) > \kthres$ or if
		$d(u) > \mk$. If $\sclk(u) \leq \frac{1}{4}\kthres$ and $d(u) \leq \mk$, then we say that $u$ is \textsf{\unpopular}.
	\end{definition}

Note that a vertex $u$ may be neither \popular\ nor \unpopular. This is the case  if
$d(u) \leq \mk$ and $\frac{1}{4}\kthres < \sclk(u) \leq \kthres$.

\smallskip
The following claim, whose proof follows directly from Definition~\ref{def:popular-v}, shows that we can ignore cliques that do not contain \unpopular\ vertices.

\begin{claim}\label{clm:u.b. pop}
If $\om \in\mrange$ and $\onk > \clk/4$, then at most $\oeps\clk$ $k$-cliques consist solely of vertices
that are not \unpopular.
\end{claim}
\ifnum\conf=0	
\begin{proof}

	For every  vertex $u$ that is not \unpopular, either \[\sclk(u) > \frac{1}{4}\kthres \text{ \; or \;\;} d(u)>  \mk.\]
	Recall that $\sum_{u \in V} c_k(u) = k \cdot \clk.$
	Therefore, if $\onk \geq \clk/4$,
	then there are at most
	\[\frac{k\cdot \clk}{\frac{1}{4}\kthres} \leq (\oeps \clk)^{1/k}/2\]
	vertices of the former type, and if  $\om \geq (1-\eps)m$, then there are
	at most
	\[\frac{m}{\mk} \leq (\oeps \clk)^{1/k}/2\]
	of the latter type. This implies that there are at most $(\oeps \clk)^{1/k}$ vertices that are not \unpopular, and it follows that the number of $k$-cliques for which  all of their vertices are not \unpopular\ is at most  $\oeps \clk$.
\end{proof}
\fi

	\begin{definition}[An appropriate partition] \label{def:appropriate}
		We say that a partition $P=(V_0,V_1)$ of $V$ is \textsf{appropriate} if every \unpopular\ vertex (as defined in Definition~\ref{def:popular-v}) is in $V_0$ and every \popular\ vertex is in $V_1$ (and any other vertex can be either in $V_0$ or $V_1$).
	\end{definition}
	
	We next specify the assignment of cliques to vertices.
	\begin{definition}[Assigning cliques] \label{def:assign}
     Fix a partition $P = (V_0, V_1)$.
     \begin{asparaitem}
        \item Assignment of cliques: We \textsf{assign} each $k$-clique $K=\{u_1, \ldots, u_k\}$ to the vertex $u_i$ that is the first (according to $\prec$) vertex of $K$ in $V_0$. If all of $K$'s vertices are in $V_1$, then $K$ is \textup{not assigned} to any vertex.
	    \item $\ap(u), \ap(S)$: We denote the number of $k$-cliques assigned to $u$ (for this $P$) by $\ap(u)$.
         For a set $S$ of vertices, $\ap(S)=\sum_{u \in S}\ap(u)$.
     \end{asparaitem}
	\end{definition}
	
The following is a corollary of Claim~\ref{clm:u.b. pop}, Definition~\ref{def:appropriate} and Definition~\ref{def:assign}.	
\begin{corollary}\label{clm:wt(V)}
		For every partition $P=(V_0,V_1)$ of $V$ it holds that $\ap(V) \leq \clk$. Furthermore,
		if $P=(V_0,V_1)$ is appropriate,
 $\om \geq (1-\eps)m$ and $\onk \geq \clk/4$, then
\confeqn{\ap(V) \in [(1-\oeps)\clk,\clk]\;.}
\end{corollary}

Another distinction between types of vertices that will play a central role in our analysis is the following.	
		\begin{definition}[High-degree and low-degree vertices]\label{def:high-low}
			We say that a vertex $u$ is a \textsf{high-degree} vertex if $d(u) > \thres$ and otherwise we say it is a \textsf{low-degree} vertex.
		\end{definition}

	\subsection{The main algorithm and the procedures it uses}\label{subsec:main-alg}

    In this subsection we present our main algorithm and the corresponding main theorem.
\ifnum\conf=0
    Our algorithm invokes several procedures, which are provided in the following subsections.
\else
    Our algorithm invokes several procedures.
\fi
   Here we shortly describe all procedures and state the main claim regarding each of them.
 Building on these claims we give a proof sketch of the main theorem (the complete proof appears in
 \ifnum\conf=0
 Subsection~\ref{subsec:thm-correct}).
 \else
  Subsection 3.6 of the full version).
 \fi

    \paragraph{\textsf{Approximate-cliques}.} This is the main algorithm, and it is provided in Figure~\ref{fig:main-alg}.
The algorithm begins by constructing two random multisets, $S$ and $T$. The multiset $S$ is obtained by simply
selecting vertices uniformly (independently) at random. The multiset $T$ is constructed by a procedure  \link{Sample-degrees-typical}.
We show that with high probability, $S$ and $T$ have certain desired properties
(where the correctness of subsequent steps of the algorithm relies on these properties).

In Step~\ref{step:q_loop}, the algorithm calls two procedures: \link{Sample-a-clique}
 and \link{Is-\popular}.
The heart of the algorithm is the procedure \link{Sample-a-clique} that
either returns a
$k$-clique $\{u,v,w_1,\dots,w_{k-2}\}$ where $u\in S$ or  returns \emph{fail}.
The  procedure \textsf{Is-\popular} distinguishes between \popular\ and \unpopular\ vertices
(as defined in Definition~\ref{def:popular-v}). It is used in order to decide for each
 $k$-clique that is output in the previous step, whether it  is assigned to $u$
 (as defined in Definition~\ref{def:assign}).

Note that if the sample size $s$ (defined in Step~\ref{step:S}) is larger than $n$, then the algorithm can simply set $S=V$.
Similarly, if  $q$ (the number of iterations in
Step~\ref{step:q_loop}) is larger than $\om$, then the algorithm can query upfront all edges incident to $S$ and
their neighbors so that it never performs more than $\min\{m,\om\}$ queries. (If it views more than $\om$ edges, then it can abort.) Finally, we may assume that $\eps>1/{\om^{k/2}}$, since otherwise we are required to output the exact number of $k$-cliques in the graph (recall that by Claim~\ref{clm:u.b. nk}, $\clk<m^{k/2}$), and thus can simply invoke the exact enumeration algorithm of \cite{ChNi85}.

\confFigure{{\bf Approximate-cliques$\;(n,k,\om,\onk, \eps,\delta)$} \label{Approximate-cliques} 
				\smallskip
				\begin{compactenum}
					\item Let $\oeps=\eps/5$  and  $\odelta=\delta/4.$
                           \label{step:set-eps}

					\item Let $S$ be a multiset of $s=\sets$ vertices chosen uniformly 
                       at random.  \label{step:S}
					\item Query the degree of each vertex in $S$ and
					set up a data structure $D(S)$ that supports sampling a uniform edge in $E(S)$ in constant time.
                   \item  Invoke \link{Sample-degrees-typical}$(n,k,\om,\onk,\oeps,\odelta)$. If the procedure
                     returned \emph{fail}, then \textbf{return} \emph{fail}. Otherwise, let $(T,m(T), D(T))$ be its output. \label{step:sample-sets}

					\item For $i=1$ to $q=\setq$ do: \label{step:q_loop}
					\begin{compactenum}
						\item Invoke \link{Sample-a-clique}$(S,T,m(S),m(T),D(S),D(T), k,\om).$ \label{step:sample-clique}
						\item If the procedure returned \emph{fail} then set $\chi_i=0.$ Otherwise, let
						 $K_i=(u_i,v_i, w_{i,1}, \ldots, w_{i,{k-2}})$ be the $k$-tuple returned and do the following.
						    \begin{compactenum}			
							   \item Query the degree and invoke \link{Is-\popular}$(x,T,m(T),D(T),\om,\onk,n,k,\oeps,\odelta)$ on each vertex $x \in K_i$.\label{step:is-pop}
							   \item If  $u_i$ is the first vertex (according to $\prec$) in $K_i$
 for which \textsf{Is-\popular} returned \emph{\unpopular}, then set $\chi_i=1$.
							 Otherwise, set $\chi_i=0$. \label{step:chi_i}
						   \end{compactenum}
					\end{compactenum}
					\item
\textbf{Return}
                      $\hnk=\frac{m(S) (\thres)^{k-2}}{(k-2)!\cdot (s/n)}\cdot \frac{1}{q}\sum_{i=1}^q\chi_i\;$. \label{step:hnk}
 \label{step:chi}
				\end{compactenum}}
	{The main algorithm for computing a $(1\pm \eps)$-estimate of the number of $k$-cliques in a graph (given a constant factor estimate of this number). \label{fig:main-alg}}

The main theorem of our paper is the following (where the second item in the theorem is used by the geometric search algorithm for $\clk$).

	\begin{theorem}\label{thm:correct}
\sloppy
		Consider an invocation of
Algorithm \link{Approximate-cliques}$(n,k,\om,\onk,\eps,\delta)$.
\begin{enumerate}
\item If $\om\in\mrange$ and $\onk \in \krange$, then with probability at least $1-\delta$, \link{Approximate-cliques}
returns a value $\hnk$ such that
		 \confeqn{\hnk \in (1\pm \eps) \cdot \clk\;.}
	\item If $\om \in \mrange$ and $\onk > \clk$, then with probability at least $\eps/4$, \link{Approximate-cliques}
	returns a value $\hnk$ such
	\confeqn{\hnk  \leq (1+ \eps) \cdot \clk\;.}
	
		\item If $\onk\leq \om^{k/2}$, then the expected query complexity and running time of \link{Approximate-cliques} are
$O\left(\frac{n}{\onk^{1/k}}+ \frac{\max\{m,\om\}\cdot \om^{(k-2)/2}}{\onk} \cdot \frac{\clk}{\onk}\right)
\cdot\poly(\log n,1/\eps,\log(1/\delta),k)$.\footnote{In the second additive term there is actually a dependence on $2^k/(k-2)!$, which we ignored for the sake of simplicity.}
The number of queries is always upper bounded by
$O\left(\frac{n}{\onk^{1/k}}\right)\cdot\poly(\log n,1/\eps,\log(1/\delta),k) + \min\{m,\om\}$.
\label{item:thm-running}
\end{enumerate}
\end{theorem}
	
\paragraph{\link{Sample-degrees-typical}.} This procedure is described in 
\ifnum\conf = 0
Figure~\ref{fig:samp-sets}.
\else
Figure 4 in the full version.
\fi
Its goal is to output a degrees-typical multiset, which is defined below. The procedure itself
is quite simple; it repeats the process of sampling a uniform random multiset for a sufficient number of times in order to achieve this condition with high probability.

		\begin{definition}\label{def:T-typical}
		We say that a multiset $T$ of size $t$ is \textsf{degrees-typical} if
$m(T)\leq \frac{t}{n}\cdot 4\om$ and
for every high-degree vertex $w \in V$,
\ifnum\conf=0
\[\dm_T(w) \in \left(1\pm \frac{\oeps}{k}\right)\cdot \frac{t}{n}\cdot d(w)\;.\]
\else
it holds that 	
$\dm_T(w) \in \left(1\pm ({\oeps}/{k})\right)\cdot \frac{t}{n}\cdot d(w)$.
\fi
		\end{definition}

In 
\ifnum\conf=0
Subsection~\ref{subsec:T}
\else
Subsection 3.5 of the full version 
\fi
we prove the following lemma regarding the correctness and running time of the procedure \link{Sample-degrees-typical}.
\begin{lemma}\label{lem:T-typical}
\sloppy
Consider an invocation of \link{Sample-degrees-typical}$(n,k,\om,\onk,\oeps,\odelta)$.
The procedure either returns \emph{fail}, or returns a multiset $T$ together with $m(T)$
and a data structure $D(T)$ that supports selecting a uniform edge in $E(T)$ in time $O(1)$.

\sloppy
Let $\gamma=\setgamma$.
 If $\om\in \mrange$,
then with probability at least $1-\gamma$,
the procedure returns a triple $(T,m(T), D(T))$ such that the multiset $T$ is degrees-typical.

The running time of \link{Sample-degrees-typical}  
is
$O\left(
\frac{n}{\sqrt{m}}\cdot \frac{k^2\cdot \log(n/\gamma)\cdot\log(1/\gamma)}{\oeps^2 }\right)
\;.$
\end{lemma}

\paragraph{\link{Sample-a-clique}.}
As mentioned earlier, this is the most important and novel
aspect of our algorithm. Given any multiset of vertices $S$, the procedure \link{Sample-a-clique} produces
cliques
incident to $S$ with roughly uniform probability.
The procedure is given in
Figure~\ref{fig:sample-a-clique}.

\smallskip\noindent
\begin{definition}\label{def:mC}
\sloppy
Let $\mC(S)$ denote the set of $k$-tuples $(u,v,w_1,\dots,w_{k-2})$ that have the following properties:
(1) the subgraph induced by $\{u,v,w_1,\dots,w_{k-2}\}$ is a $k$-clique;  (2) $u \in S$; (3) $v\prec w_j$
for every $j \in [k-2]$.
\end{definition}
By the above definition, each clique containing a vertex $u\in S$ is associated with the edge $(u,v)$ in the
clique such that $v\prec w$ for every other vertex $w\neq u$ in the clique, and
the clique has exactly $(k-2)!$ corresponding tuples in $\mC(S)$.
In 
\ifnum\conf=0
Subsection~\ref{subsec:sample-clique}
\else
Subsection 3.3 of the  full version,
\fi
we prove the following lemma regarding the correctness and running time of the procedure \link{Sample-a-clique}.
\begin{lemma}\label{lem:sample-clique}
Let $T$ be a degrees-typical multiset and let $S$ be any multiset.
For any fixed $k$-tuple $K \in \mC(S)$, the probability that
an invocation of \link{Sample-a-clique}$(S,T,m(S),m(T),D(S),D(T),k, \om)$
returns $K$ is
in  $(1\pm \oeps) \cdot \frac{ 1}{m(S)\cdot (\thres)^{k-2}}$.
%
%
\; The running time of \link{Sample-a-clique} is $O(k^2)$.
		\end{lemma}
			
\paragraph{\link{Is-\popular}.} This procedure decides if a vertex $u$ is \popular\ or not.
This is done by multiple independent invocations of \link{Sample-a-clique} for the set $S = \{u\}$.
The procedure is provided in 
	\ifnum\conf=1
Figure 3 in 	Subsection 3.4 of the  full version,
	\else
Figure~\ref{fig:is-popular} in 	Subsection~\ref{subsec:Is-popular},
	\fi
where we also prove the following lemma regarding the correctness and running time of the procedure \link{Is-\popular}.
\begin{lemma}\label{lem:is-popular}
	\sloppy
Let $T$ be a degrees-typical multiset.
If  \link{Is-\popular}$(u,T,m(T),D(T),\om,\onk,n,k,\oeps,\odelta)$ is invoked
 with a \popular\ vertex $u$, then with probability at least $1-\odelta/n$, the procedure returns \emph{\popular}, and if $u$ is \unpopular,  then with probability at least $1-\odelta/n$, the procedure returns \emph{\unpopular}.
				\; The running time of \textsf{Is-\popular} is
$O\left(\frac{\om^{k/2}}{\onk}\cdot\frac{k\cdot 2^k\cdot \log(n/\odelta)}{(k-2)!\cdot\oeps^{2-1/k}}\right)$.
\end{lemma}

\paragraph{Proof sketch of the first item in Theorem~\ref{thm:correct}.}
The full proof of Theorem~\ref{thm:correct} appears in 
	\ifnum\conf=1
	Subsection 3.6 of the  full version.
	\else
	Subsection~\ref{subsec:thm-correct}.
	\fi
Here we provide a proof sketch for the case that $\onk \in \krange$, relying on
Lemmas~\ref{lem:T-typical}--\ref{lem:is-popular}.

By the first premise of this item, $\om \in \mrange$. Hence, by Lemma~\ref{lem:T-typical}, with high probability
the multiset $T$ is degrees-typical.
From this point on we condition on this event.

Consider (as a thought experiment) invoking \link{Is-\popular} on {\em all\/} vertices with the degrees-typical multiset $T$. Based on these invocations, we can define a partition $\Pip = (V_0,V_1)$, where $V_0$ contains all vertices for which \link{Is-\popular}
returns \emph{\unpopular} and $V_1$ contains all vertices for which \link{Is-\popular}
returns \emph{\popular}. By Lemma~\ref{lem:is-popular} and a union bound over all vertices, we get that
with probability at least $1-\odelta$,
$\Pip$ is an appropriate partition (as defined in Definition~\ref{def:appropriate}).
Conditioned on $\Pip$ being appropriate (and using our assumptions on $\om$ and $\onk$), by Corollary~\ref{clm:wt(V)} we have that $\alpha_{\Pip}(V) \in [(1-\oeps)\clk,\clk]$.

Now consider the selection of the multiset $S$. We show that the size $s$ of this sample ensures that with high probability $\alpha_{\Pip}(S)$ is close to its expected value, $\frac{s}{n}\cdot \alpha_{\Pip}(V)$,
so that $\alpha_{\Pip}(S) \in (1\pm\oeps)^2 \cdot \frac{s}{n}\cdot \clk$. We condition on this event as well.
Since $T$ is degree-typical, by Lemma~\ref{lem:sample-clique}, whenever we invoke \link{Sample-a-clique} (in Step~\ref{step:sample-clique}) it returns each $k$-tuple in $\mC(S)$ with probability approximately $\frac{1}{m(S)\cdot (\thres)^{k-2}}$.
Observe that if \link{Sample-a-clique} returns a $k$-tuple
$K_i = (u_i,v_i,w_{i,1},\dots,w_{i,k-2})$, then in Steps~\ref{step:is-pop} and~\ref{step:chi_i} the
algorithm determines whether the corresponding $k$-clique is assigned to $u_i$ according to $\Pip$
and sets $\chi_i$ to $1$. Therefore,
$\Pr[\chi_i=1]$ is approximately $\frac{(k-2)!\cdot \alpha_{\Pip}(S)}{m(S)\cdot (\thres)^{k-2}}$
(recall that for each $k$-clique that a vertex $u \in S$ participates in, there are $(k-2)!$ corresponding
$k$-tuples in $\mC(S)$).
By the setting of the number of invocations, $q$, of \link{Sample-a-clique}, with high probability, the sum of the $\chi_i$'s is close to its expected value, and the output of the algorithm is as claimed.

\subsection{Sampling a clique} \label{subsec:sample-clique}

In this subsection we provide the procedure \link{Sample-a-clique}
and prove Lemma~\ref{lem:sample-clique}. The procedure first samples a uniform edge $(u,v)$ in $E(S)$. It then tries to construct a $k$-clique that this edge participates in  by selecting $k-2$ additional vertices.
More precisely, it tries to construct a $k$-tuple $(u,v,w_1,\dots,w_{k-2})$ in $\mC(S)$.
Recall that such a $k$-tuple
satisfies: $u \in S$, $\{u,v,w_1,\dots,w_{k-2}\}$ induces a $k$-clique, and $v\prec w_j$
for each $j \in [k-2]$.
To this end the procedure repeats the following $k-2$ times. If $v$ is a low-degree vertex, then it selects
a uniform neighbor $w$ of $v$, and maintains it with probability $d(v)/\thres$.
If $v$ is a high-degree vertex, then the procedure tries to sample a random high-degree vertex.
It does so by first sampling a uniform edge  $(x,y) \in E(T)$ and if $y$ is a high-degree neighbor of $v$, performing rejection sampling according to the degree of $y$.
We prove that conditioned on $T$ being degrees-typical, we obtain each $k$-tuple in $\mC(S)$ with
almost equal probability.

\vspace{2mm}
\confFigure{
			{\bf Sample-a-clique$\;(S,T,m(S),m(T),D(S),D(T),k,\om)$} \label{Sample-a-clique}
			\smallskip
			\begin{compactenum}
				\item Sample a uniform edge $e=(u,v)$ in $E(S)$ (using the data structure $D(S)$).

				\item For $j=1$ to $k-2$ do:
				\begin{compactenum}
						
				\item If $d(v)\leq \thres$ ($v$ is a low-degree vertex), then:\label{step:sample-low}
				\begin{compactenum}
                    \item Uniformly select a neighbor $w_j$ of $v$.
                    \item Keep $w_j$ with probability $\frac{d(v)}{\thres}$ and     with probability $1- \frac{d(v)}{\thres}$ \textbf{return} \emph{fail}.
				\end{compactenum}
				\item Else ($v$ is a high-degree vertex):
				\begin{compactenum}
					\item Sample a random edge $(x,y)$ in $E(T)$ (using the data structure $D(T)$). \label{step:sample_xy}
                    \item If $d(y)  \leq \thres$, \textbf{return} \emph{fail}.
                    \item With probability $ \frac{m(T)}{d(y)\cdot \frac{t}{n}\cdot \thres}$ set $w_j = y$, otherwise, \textbf{return} \emph{fail}.\label{step:return y} \label{step:sample-high}

				\end{compactenum}	
			\end{compactenum}
				\item For every pair of vertices in $\{u,v, w_{1}, \ldots, w_{{k-2}}\}$ query if there is an edge between the two vertices.\label{step:sample-in-E(T)}
				\item If the subgraph induced by $\{ u,v,w_{1}, \ldots, w_{{k-2}} \}$
 is a $k$-clique and  $v \prec  w_{j}$ for every $j \in [k-2]$ then
				\textbf{return} $K=(u,v,w_{1}, \ldots, w_{{k-2}} )$. Otherwise \textbf{return} \emph{fail}. \label{step:w_i_j}
			\end{compactenum}	
	}
{A procedure for sampling a $k$-clique incident to $S$ with almost uniform probability.\label{fig:sample-a-clique}} 

\smallskip
\begin{proofof}{Lemma~\ref{lem:sample-clique}}
Let $(a,b,z_1,\dots,z_{k-2})$ be a $k$-tuple in $\mC(S)$. Recall that by the definition of $\mC(S)$ we have
that $a\in S$, the subgraph induced by $\{a,b,z_1,\dots,z_{k-2}\}$ is a $k$-clique and $b\prec z_j$ for every $j \in [k-2]$.
If \link{Sample-a-clique} does not return \emph{fail}, then its output is a tuple
$(u,v,w_1,\ldots,w_{k-2})$ in $\mC(S)$.
The probability that the procedure returns the tuple $(a,b,z_1, \ldots, z_k)$ is
\ifnum\conf=1
\small
\begin{align*}
& \Pr\Big[(u,v)=(a,b) \text{ and } \forall j \in [k-2]\;, w_j=z_j \; \Big] 
\\ &= \Pr[(u,v)=(a,b)] \cdot \Pr\Big[\forall j \in [k-2] \;, w_j=z_j  \;|\; (u,v)=(a,b) \Big]\;.
\end{align*}
\normalsize	
\else
\begin{eqnarray*}
\lefteqn{\Pr\Big[(u,v)=(a,b) \text{ and } \forall j \in [k-2]\;, w_j=z_j \; \Big]}\\
  &=&\Pr[(u,v)=(a,b)] \cdot \Pr\Big[\forall j \in [k-2] \;, w_j=z_j  \;|\; (u,v)=(a,b) \Big]\;.
\end{eqnarray*}
\fi
Clearly, $\Pr[(u,v)=(a,b)] = \frac{1}{m(S)}$, so it remains to compute the probability that
$w_j = z_j$ for each $j\in [k-2]$, conditioned on $(u,v) = (a,b)$.

If  $b=v$ is a low-degree vertex, then the vertices $w_{1}, \ldots, w_{{k-2}}$ are
uniformly selected random neighbors of $v$.
For each $j \in [k-2]$, the probability that $w_j=z_j$ and that
the procedure did not return \emph{fail} due to rejection sampling   is
 $\frac{1}{d(v)} \cdot \frac{d(v)}{\thres} = \frac{1}{\thres}$.  
 Therefore, $\Pr\Big[\forall j \in [k-2] \;, w_j=z_j  \;|\; (u,v)=(a,b) \Big]
                        = 1/(\thres)^{k-2}$.

Otherwise ($b$ is a high-degree vertex), since $b\prec z_j$ for each $j \in [k-2]$,
  the vertices $z_{1}, \ldots, {z_{k-2}}$ are also high-degree vertices. In this case (conditioned on $(u,v)=(a,b)$), the procedure tries to sample $k-2$ high-degree vertices
by sampling edges originating from the vertices of $T$. We next prove that, in Step~\ref{step:sample-high}, any specific high-degree vertex is sampled with probability in $\left(1\pm \frac{\oeps}{k}\right)\frac{1}{\thres}$.

Since $T$ is degrees-typical (as defined in Definition~\ref{def:T-typical}), $m(T) \leq \frac{t}{n} \cdot 4\om$.
This implies that for every high-degree vertex $z$,
\[\frac{m(T)}{d(z)\cdot \frac{t}{n}\cdot \thres} \leq \frac{\frac{t}{n}\cdot4\om}{2\sqrt{\om}\cdot \frac{t}{n}\cdot\thres} = 1.\]
Thus, Step~\ref{step:w_i_j} is valid.
Since $T$ is degrees-typical, for every high-degree vertex $z$,
\confeqn{\dm_T(z) \in \left(1\pm \frac{\oeps}{k}\right)\cdot \frac{t}{n}\cdot d(z)\;. 
	}

For any vertex $y$, the probability of obtaining an edge in Step~\ref{step:sample_xy} with $y$
as an endpoint is $d_T(y)/m(T)$. Therefore, for
 each $j \in [k-2]$,
the probability that $w_{j}=z_j$ in Step~\ref{step:w_i_j} is
\confeqn{\frac{\dm_T(z_j)}{m(T)}\cdot \frac{m(T)}{d(z_j)\cdot \frac{t}{n}\cdot \thres}  \in  \left(1\pm \frac{\oeps}{k}\right)\cdot \frac{1}{\thres}\;.}
It follows that  the probability that the procedure returns any specific $k$-tuple in $\mC(S)$ is
$(1\pm\oeps)\cdot \frac{1}{m(S)(\thres)^{k-2}}$.

\sloppy
 It remains to bound the running time of the procedure. Given the data structure $D(S)$, it takes time $O(1)$ to sample an edge in $E(S)$,
and similarly it takes time $O(1)$ to sample an edge in $D(T)$.
The procedure samples a single edge in $E(S)$ and possibly $k-2$ edges in $E(T)$. Adding the time to perform queries on all pairs of vertices in
$\{u,v,w_1,\dots,w_{k-2}\}$ (in addition to a degree query on each of these vertices),
the total running time is $O(k^2)$.
\end{proofof}

\ifnum\conf=0

\subsection{Is-\popular}\label{subsec:Is-popular}
The procedure \link{Is-\popular} determines (with high success probability)
whether a given vertex $u$ is \popular\ or \unpopular.
For vertices that are neither \popular\ nor \unpopular, it can answer arbitrarily.
Recall that
by Definition~\ref{def:popular-v}, the distinction between a \popular\  $u$ and a \unpopular\ $u$
involves bounds on both
$d(u)$ and $c_k(u)$. These will be critical in bounding the running time of \link{Is-\popular}.
The procedure basically invokes \link{Sample-a-clique}
repeatedly to check if $c_k(u)$ is larger than the
specified threshold.
	
\confFigure{
				{\bf Is-\popular$\;(u,T,m(T),D(T),\om,\onk,n,k,\oeps,\odelta)$} \label{Is-\popular} 
				\smallskip
				\begin{compactenum}
					\item Query the degree of $u$ and if $d(u) > \mk$ then \textbf{return} \emph{\popular}.
					\item For $i=1$ to $r=\qpopular$ do: \label{step:is-popular-for}
					\begin{compactenum}
						\item Invoke \textsf{Sample-a-clique$(\{u\},T,d(u),m(T),D(\{u\}),D(T),k, \om)$}.
						\item If a $k$-tuple (corresponding to a $k$-clique) was returned, then set $\chi_i=1$. Otherwise (the procedure returned \emph{fail}), set $\chi_i=0$. \label{step:pop-chi}
					\end{compactenum}
					\item Let $\shnk(u)=\frac{d(u)\cdot \thres}{r}\cdot \sum_{i=1}^r \chi_i\;.$
					\item If $\shnk(u)\geq \frac{1}{2}(\kthres)$, then \textbf{return} \emph{\popular},
                     otherwise, \textbf{return} \emph{\unpopular}.
				\end{compactenum}
			}{A procedure for determining with high probability whether a given vertex is \popular\ or not. \label{fig:is-popular}}


	\begin{proofof}{Lemma~\ref{lem:is-popular}}
\sloppy
Consider any fixed vertex $u$.
Let $\chi=\frac{1}{r}\sum_{i=1}^r\chi_i$, where $\chi_1, \ldots,\chi_r$ are as defined in Step~\ref{step:pop-chi} of \textsf{Is-\popular}.	
Note that $S=\{u\}$, $m(S)= m(\{u\}) = d(u)$ and $|\mC(S)| = \mC(\{u\}) =(k-2)!\cdot \sclk(u)$.
By Lemma~\ref{lem:sample-clique} and the assumption that $T$ is degrees-typical,
$\EX[\chi] \in (1\pm \oeps)\cdot \frac{(k-2)! \cdot \sclk(u)}{d(u) \cdot (\thres)^{k-2}}\;$.

		First consider the case that $u$ is \popular. By Definition~\ref{def:popular-v}, either $d(u)> \mk$ or  $\sclk(u) \geq \kthres$ (or both).
		Clearly, if $d(u)> \mk$, then the procedure returns \emph{\popular}.
		Hence, assume that $\sclk(u) \geq \kthres$.  By Chernoff's inequality and the setting of \[r=\qpopular\] in
  Step~\ref{step:is-popular-for} of the procedure,
\ifnum\stoc=0
		\begin{align*}
		 \Pr\left[ \frac{1}{r}\sum_{i=1}^r \chi_i < (1-\oeps)\cdot \EX[\chi]\right] &< \exp\left(-\frac{\oeps^{\hspace{1pt}2} \cdot \EX[\chi] \cdot r}{3}\right)  \\
		 &< \exp\left(-\frac{\oeps^{\hspace{1pt}2} \cdot  \frac{(1-\oeps)\cdot (k-2)! \cdot \sclk(u)}{d(u) \cdot (\thres)^{k-2}} \cdot \qpopular}{3}\right) \\
		 &\leq \frac{\odelta}{n}\;.
		 \end{align*}
\else
\[
\Pr\left[ \frac{1}{r}\sum_{i=1}^q \chi_i < (1-\oeps)\cdot \EX[\chi]\right] \;<\; \exp\left(-\frac{\oeps^{\hspace{1pt}2} \cdot \EX[\chi] \cdot r}{3}\right) \;\leq\; \frac{\odelta}{n}\;.
\]
\fi
		It follows that if $T$ is typical and $u$ is \popular, then with probability at least $1-\odelta/n$,		
		$$\shnk(u) \geq (1 - \oeps)^2 \cdot \sclk(u) \geq  \frac{1}{2}\cdot \kthres,$$
causing the procedure to return \emph{\popular}.

		Now consider the case that $u$ is \unpopular. By Definition~\ref{def:popular-v}, 
		\[\sclk(u) \leq \frac{1}{4}\cdot \kthres,\] implying that \[\EX[\chi]< (1+\oeps)\cdot \frac{\frac{(k-2)!}{4}\cdot (\kthres)}{d(u)\cdot (\thres)^{k-2}}.\]
		Hence, by	Chernoff's inequality, and by the setting of $r$,
		\begin{align*}
		 \Pr\left[  \frac{1}{r}\sum\limits_{i=1}^r \chi_i > \frac{3}{2}\cdot \left(\frac{(1+\oeps)}{4}\cdot \kthres\right) \right] <
\ifnum\stoc=0
		 \exp\left(-\frac{\left( \frac{\frac{(1+\oeps)(k-2)!}{4}\cdot (\kthres)}{d(u)\cdot (\thres)^{k-2}} \right)\cdot r}{12}\right)<
	 \fi
		  \frac{\odelta}{ n}.
		 \end{align*}
		Therefore, if $T$ is degrees-typical and $u$ is \unpopular, then  with probability at least $1-\odelta/n$, the algorithm returns \emph{\unpopular}.

\sloppy
It remains to bound the running time of the procedure.
	By Lemma~\ref{lem:sample-clique}, the procedure \link{Sample-a-clique} runs in time $O(k^2)$.
	Crucially, \link{Sample-a-clique} is invoked only if $d(u) \leq \mk$. Since
     \[r=\qpopular,\] the running time of the procedure is
\[O\left(r\cdot k^2\right) =O\left(\frac{\om^{k/2}}{\onk}\cdot\frac{k\cdot 2^k \cdot \log(n/\odelta)}{(k-2)!\cdot\oeps^{2-1/k}}\right)\;,\]
as claimed.
	\end{proofof}

	\subsection{Sampling a degrees-typical set}\label{subsec:T}
In this subsection we provide the (simple) procedure \link{Sample-degrees-typical} and prove
Lemma~\ref{lem:T-typical} regarding its correctness and running time.
As in the case of the choice of the sample $S$ by \link{Approximate-cliques}, here too
if the sample size $t$ is larger than $n$, then the algorithm can simply set $T=V$.
%
%

\confFigure{
				{\bf Sample-degrees-typical$\;(n,k,\om,\onk, \oeps, \odelta)$} \label{Sample-degrees-typical}
				\smallskip
               \begin{compactenum}
					\item $\gamma=\setgamma$.
					\item For $i=1$ to $\log(2/\gamma)$ do:
					\begin{compactenum}
					\item Let $T_i$ be a multiset of $t=\sett$ vertices chosen uniformly, 
                      at random.
					\label{step:sets T}
					\item Query the degrees of all the vertices in $T_i$ and compute $m(T_i)$.
					\end{compactenum}
					\item Let $T$ be the first set $T_{i}$ such that  $m(T_{i})\leq \frac{t}{n}\cdot 4\om$. If no such set exists, then  \textbf{return} \emph{fail}. Else, set up a data structure $D(T)$
that supports sampling a uniform edge in $E(T)$ in constant time.
\label{step:T}
                   \item \textbf{Return} $(T,m(T),D(T))$.
				\end{compactenum}
		}
{The procedure for sampling the multiset $T$.\label{fig:samp-sets}}

	\begin{proofof}{Lemma~\ref{lem:T-typical}}
For each iteration $i$, consider the selection of the multiset $T_i$.
For any fixed high-degree vertex $u\in V$
and for $j=1, \ldots, t$, let $\chi_j(u)=1$ if the $j\th$ vertex in $T_i$ is a neighbor of $u$ and let $\chi_j(u)=0$ otherwise. By the definition of $\chi_1(u), \ldots, \chi_t(u)$ and the premise of the lemma regarding $\om$, $\EX\left[\frac{1}{t}\sum_{j=1}^t \chi_j(u)\right] = \frac{\dm(u)}{n} \geq \frac{\sqrt{m}}{n}$.
		By Chernoff's inequality and the setting of $t=\sett$ in Step~\ref{step:T},
\ifnum\stoc=1
		\begin{align*}
		 &\Pr\left[\left| \frac{1}{t}\sum\limits_{j=1}^{t} \chi_j(u) -\frac{\dm(u)}{n} \right| > \frac{\oeps}{k}\cdot  \frac{\dm(u)}{n}\right]  \\&  \;\;\;\;\;\;\;
		 <
		 2\exp\left(-\frac{(\oeps/k)^2\cdot t\cdot \frac{\dm(u)}{n} }{3}\right)
\nonumber
< \frac{\gamma^{\hspace{1pt}2}}{2n}\;.\nonumber
		\end{align*}
		\else
				\begin{align*}
		\Pr\left[\left| \frac{1}{t}\sum\limits_{j=1}^{t} \chi_j(u) -\frac{\dm(u)}{n} \right| > \frac{\oeps}{k}\cdot  \frac{\dm(u)}{n}\right]  
		<
		2\exp\left(-\frac{(\oeps/k)^2\cdot t\cdot \frac{\dm(u)}{n} }{3}\right)
		\nonumber
		< \frac{\gamma^{\hspace{1pt}2}}{2n}\;.
		\end{align*}
		\fi
		By taking a union bound over all high-degree vertices, it holds that with probability at least $1-\gamma^{\hspace{1pt}2}/2$, for every high-degree vertex $u \in V$,
		$ \dm_{T_i}(u) = \sum\limits_{j=1}^t \chi_j(u) \in \left(1\pm (  {\oeps}/{k})\right)\cdot \frac{t}{n}\cdot \dm(u) .$
		By taking a union bound over the $\log(2/\gamma)$ sampled multisets $T_i$, it holds that with probability at least $1-\frac{1}{2}\gamma^{\hspace{1pt}2} \log(2/\gamma)> 1-\gamma/2$, for each of the multisets $T_i$ and
		for every high-degree vertex $u\in V$,
		$d_{T_i}(u) \in \left(1\pm \frac{\oeps}{k}\right) \cdot \frac{t}{n}\cdot d(u)\;.$

We now turn to bounding the probability that none of the multiset $T_i$  satisfies
$m(T_i)\leq \frac{t}{n}\cdot 4\om $.
 By the definition of $m(\cdot)$, for every $T_i$ we have that $\EX[m(T_i)]=\frac{t}{n}\cdot m$. By the assumption that $\om \in [(1-\oeps)m,m]$ and by Markov's inequality, $\Pr[m(T_i)> \frac{t}{n}\cdot 4\om ]\leq  \Pr[m(T_i)> \frac{t}{n}\cdot 2m]< \frac{1}{2}$. Hence, the probability that $m(T_i)> \frac{t}{n}\cdot 4\om $ for
 every $i = 1,\dots,\log(2/\gamma)$ is at most $\gamma/2$.

By combining the two failure probabilities, we get that with probability at least $1-\gamma$, the algorithm returns a multiset $T$ that is degrees-typical (as defined in Definition~\ref{def:T-typical}).

Finally, by a performing a preprocessing step that takes $\Theta(t)$ time  it is possible to build a data structure $D(T)$ that allow for sampling each vertex  $u \in S$ with probability proportional to $d(u)/m(T)$ (see e.g., \cite{walker1974new, walker1977efficient, marsaglia2004fast}).
This in turn implies that, using $D(T)$, it is possible to sample a uniform edge in $E(T)$ in constant time.

%
\sloppy
	The running time of \link{Sample-degrees-typical} is
\[O\left(t\cdot
		\log(1/\gamma)\right)=O\left(
		\frac{n}{\sqrt{\om}}\cdot \frac{k^2 \cdot \log(n/\gamma)\cdot\log(1/\gamma )}{\oeps^{\hspace{1pt}2} }\right) \;,\]
		and the proof is complete.
	\end{proofof}
\subsection{Proof of Theorem~\ref{thm:correct}}\label{subsec:thm-correct}

In this subsection we prove Theorem~\ref{thm:correct}. We first define the notion of a {\em cliques-typical}
multiset.
		\begin{definition}\label{def:S-typical}
			We say that a multiset $S$ of size $s$ is \textsf{cliques-typical} with respect to a partition $P$ if 	 \confeqn{\ap(S) \in (1\pm \oeps)^2 \cdot \frac{s}{n}\cdot \clk\;.}
		\end{definition}

We establish a simple claim regarding the
multiset $S$ selected by our algorithm (appropriate partitions are as defined in
Definition~\ref{def:appropriate}).

\begin{claim}\label{clm:S-typical}
	Consider the multiset $S$ sampled in Step~\ref{step:S} of Algorithm \link{Approximate-cliques}.
	For any fixed partition $P$ of $V$,
	$\EX[\ap(S)]\leq \frac{s}{n}\cdot \clk$.
	Furthermore, if $P$ is appropriate,  $\om \geq (1-\eps)m$ and $\onk \in \krange$,
then	with probability at least $1-\odelta$ (over the choice of $S$), the multiset $S$ is cliques-typical.
\end{claim}
\begin{proof}
	Recall that by Definition~\ref{def:assign}, $\ap(S)$ is the number of $k$-cliques assigned to the vertices of $S$, and that by Corollary~\ref{clm:wt(V)}, for every partition $P$, $\ap(V)  \leq \clk\;.$
	Hence, $\EX[\ap(S)] = \frac{s}{n}\cdot \ap(V) \leq \frac{s}{n}\cdot \clk$.
	
	We turn to consider the case that the partition $P$ is appropriate, $\om \geq (1-\eps)m$
and $\onk \in \krange$. By Corollary~\ref{clm:wt(V)}, $\ap(V) \in [(1-\oeps)\clk,\clk]$.
	By Definition~\ref{def:assign}, $k$-cliques are only assigned to vertices that are not \popular, implying that for every vertex $u\in V$, $\ap(u) \leq \sclk(u) \leq \kthres$.
Hence, by the multiplicative Chernoff bound and the setting of $s = \sets$,
\ifnum\stoc=1
	\begin{align*}
	&\Pr\left[ \left| \frac{1}{s}\sum_{u \in S}\ap(u) - \frac{\ap(V)}{n}\right| > \oeps \cdot \frac{\ap(V)}{n} \right] 
	\\ &  \;\;\;\;\;\;\; < 2\exp\left(-\frac{\oeps^{\hspace{1pt}2} \cdot \frac{\ap(V)}{n} \cdot s}{3 \cdot \kthres} \right)
	\\	&  \;\;\;\;\;\;\;
	\leq  2\exp\left(-\frac{\oeps^{\hspace{1pt}2+\frac{1}{k}} \cdot \frac{(1-\oeps)\clk}{n} \cdot \sets}{150 \cdot \onk^{1-1/k}} \right) <
	\odelta\;.
	\end{align*}
\else
$$
	\Pr\left[ \left| \frac{1}{s}\sum_{u \in S}\ap(u) - \frac{\ap(V)}{n}\right| > \oeps \cdot \frac{\ap(V)}{n} \right]  < 2\exp\left(-\frac{\oeps^{\hspace{1pt}2} \cdot \frac{\ap(V)}{n} \cdot s}{3 \cdot \kthres} \right)
	 < 	\odelta\;.
$$
\fi
	Therefore, if $P$ is appropriate, $\om \geq (1-\eps)m$
and $\onk  \in \krange$, with probability at least $1-\odelta$,
	\[\ap(S) \in (1\pm \oeps)\cdot \frac{s}{n} \cdot \ap(V) \in (1\pm \oeps)^2 \cdot \frac{s}{n}\cdot \clk, \]
	which implies that $S$ is cliques-typical by Definition~\ref{def:S-typical}.
\end{proof}

\subsubsection{The case {$\onk\in\krange$}}\label{subsubsec:onk-good}
We start with the first item in the theorem.
Recall that by the first premise of this item, $\om\in\mrange$.
\sloppy
By Lemma~\ref{lem:T-typical}, with probability at least $1-\setgamma\geq 1-\odelta$ the procedure \link{Sample-degrees-typical} returns a multiset $T$ that is degrees-typical. We henceforth condition on this event.

Since $T$ is degrees-typical, by Lemma~\ref{lem:is-popular}, an invocation of the procedure \link{Is-\popular} with a \unpopular\ vertex $u$ returns \emph{\unpopular} with probability at least $1-\odelta/n$, and similarly an invocation with a \popular\ vertex $u$ returns \emph{\popular} with probability at least $1-\odelta/n$.
Consider (as a thought experiment)
running the procedure \link{Is-\popular} on all the vertices in the graph, and letting $V_0$ be the set of vertices for which the procedure returned \emph{\unpopular}, $V_1=V \setminus V_0$
and $\Pip=(V_0,V_1)$.
Conditioned on the event that  $T$ is degrees-typical, by Lemma~\ref{lem:is-popular} and by taking a union bound over all the vertices in $V$,
with probability at least $1-\odelta$ the partition $\Pip$ is appropriate
(as defined in Definition~\ref{def:appropriate}). Suppose we fix
the random coins used by \link{Is-\popular} on all vertices to a setting that indeed induces
an appropriate partition $\Pip$, and assume that all calls made by the algorithm to \link{Is-\popular}
are answered consistently with $\Pip$. (The probability that $\Pip$ is not appropriate is taken into
account in the total failure probability of the algorithm.)
Conditioned on $\Pip$ being appropriate (and using our assumptions on $\om$ and $\onk$), by Corollary~\ref{clm:wt(V)} we have that $\alpha_{\Pip}(V) \in [(1-\oeps)\clk,\clk]$.

Now consider the multiset $S$ sampled in Step~\ref{step:S} of the algorithm.
By Claim~\ref{clm:S-typical}, with probability at least $1-\odelta$, $S$ is cliques-typical with respect to $\Pip$. That is, $\alpha_{\Pip}(S) \in (1\pm \oeps)^2\cdot \frac{s}{n} \cdot \clk$.
We condition on this event as well.
Since $T$ is degree-typical, by Lemma~\ref{lem:sample-clique}, whenever we invoke \link{Sample-a-clique} (in Step~\ref{step:sample-clique}) it returns each $k$-tuple in $\mC(S)$ with probability in
$(1\pm\oeps) \cdot \frac{1}{m(S)\cdot (\thres)^{k-2}}$.

By the description of Steps~\ref{step:is-pop} and~\ref{step:chi_i}, $\chi_i$ is set to $1$ only if the procedure \link{Sample-a-clique} returns a $k$-tuple $K=(u,v,w_1,\ldots,w_{k-2})$ in $\mC(S)$ that is assigned to $u$ according to $\Pip$.
For each $k$-clique assigned to a vertex $u\in S$, the number of corresponding $k$-tuples in $\mC(S)$
is $(k-2)!$.
Therefore, for a cliques-typical multiset $S$,
\ifnum\stoc=1
\begin{align}
\EX[\chi_i] &\in (1\pm \oeps)\cdot \frac{\alpha_{\Pip}(S)\cdot (k-2)!}{m(S)\cdot (\thres)^{k-2}} \nonumber
\\& \in (1\pm \oeps)^3\cdot \frac{\frac{s}{n}\cdot \clk \cdot (k-2)!}{m(S)\cdot (\thres)^{k-2}}\;.\label{eqn:exp-xhi-assum}
\end{align}
\else
\[
\EX[\chi_i] \in (1\pm \oeps)\cdot \frac{\alpha_{\Pip}(S)\cdot (k-2)!}{m(S)\cdot (\thres)^{k-2}} \nonumber
\in (1\pm \oeps)^3\cdot \frac{\frac{s}{n}\cdot \clk \cdot (k-2)!}{m(S)\cdot (\thres)^{k-2}}\;.\label{eqn:exp-xhi-assum} \numberthis
\]
\fi

By the multiplicative Chernoff bound and by the setting of $q=\setq$ in Step~\ref{step:q_loop} of the algorithm,
\ifnum\stoc=0
\begin{align*}
& \Pr\left[\left| \chi-\EX[\chi] \right|> \oeps \cdot \EX[\chi] \right] <
2\exp\left(-\frac{\oeps^{\hspace{1pt}2} \cdot \EX[\chi] \cdot q}{3}\right)
\\& \;\;\;
<  2\exp\left(-\frac{\oeps^{\hspace{1pt}2} \cdot \frac{(1-\oeps)^{3}\cdot (k-2)! \cdot \clk \cdot \frac{s}{n}}{m(S) \cdot \left(\thres\right)^{k-2}} \cdot \setq}{3}\right) \
< \odelta\;.
\end{align*}
\else
\begin{align*}
& \Pr\left[\left| \chi-\EX[\chi] \right|> \oeps \cdot \EX[\chi] \right] <
2\exp\left(-\frac{\oeps^{\hspace{1pt}2} \cdot \EX[\chi] \cdot q}{3}\right)
< \odelta\;.
\end{align*}
\fi
Therefore, if $\om \in \mrange$, $\onk\in\krange$, $T$ is degrees-typical, $\Pip$ is appropriate and $S$ is cliques-typical, then with probability at least $1-\odelta$,		
$$\chi \in (1\pm \oeps)^4\cdot \frac{\frac{s}{n}\cdot \clk \cdot (k-2)!}{m(S)\cdot (\thres)^{k-2}}\;
\implies \hnk \in (1\pm \oeps)^{4} \cdot \clk \in (1\pm\eps)\clk\;, $$
where we have used the fact that $\oeps = \eps/5$.
By taking a union bound over all bad events, 
$\clk \in (1\pm\eps)\clk$
with probability at least $1-4\odelta>1-\delta $ (since $\odelta = \delta/4$).

\subsubsection{The case {${\onk>\clk}$}}\label{subsubsec:onk-large}
We now prove the second item of the theorem. As in the first item, since $\om \in \mrange$, by Lemma~\ref{lem:T-typical} with probability at least $1-\setgamma$  the multiset $T$ is degrees-typical. Conditioned on $T$ being degrees-typical, an invocation of the procedure \link{Sample-a-clique} returns each $k$-tuple in $\mC(S)$ with probability in $(1\pm \eps)\frac{1}{m(S)\cdot (\thres)^{k-2}}$.


In this case, since $\onk> \clk$,
the 
setting of $r$ in the procedure \link{Is-\popular} is not sufficiently large to ensure that
the procedure is accurate, and that the partition $\Pip$
is appropriate (with high probability).
Similarly, $\alpha_{\Pip}(S)$ might not be close to its expected value.
Therefore, we only use an upper bound on the expected value of $\alpha_{\Pip}(S)$, as explained next.

Consider the following random variables. For a multiset $S$ and for each $i \in [q]$
let $\chi_i(S)$  denote the random variable $\chi_i$ conditioned on $S$ (where $q$ and $\chi_i$ are as defined in Step~\ref{step:q_loop} of the algorithm).
Let $Y_i(S) = \frac{m(S)\cdot (\thres)^{k-2} }{(k-2)!(s/n)}\cdot \chi_i(S)$,
 and let $\hnk(S)$ denote the value of the random variable $\hnk$ (which is the output of the algorithm),
 conditioned on $S$.
    We use the notation \emph{S-c} to denote the random coins of the procedure \link{Sample-a-clique}. By the above discussion and by the setting of $\hnk$ in Step~\ref{step:hnk},  if $T$ is degrees-typical, then
\begin{align*}
\EX\left[\hnk\right] &= \EX_{S}\left[  \EX_{\emph{S-c}}\left[\hnk(S)\right] \right]
\\ &=
\EX_{S}\left[ \EX_{\emph{S-c}}\left[
\frac{m(S)\cdot (\thres)^{k-2}}{(k-2)!\cdot \frac{s}{n}} \cdot \frac{1}{q}\sum\limits_{i=1}^{q}\chi_i(S)
\right] \right]
\\ &=
\EX_{S}\left[ \EX_{\emph{S-c}}\left[ \frac{1}{q} \sum_{i=1}^q Y_i(S)\right] \right]
\\&=\EX_{S}\left[ \EX_{\emph{S-c}}\left[ Y_1(S)\right] \right] \numberthis \label{eqn:exp_hnk} \;.
\end{align*}
The last equality holds simply because all the $Y_i(S)$ variables are equally distributed (for each fixed $S$).
We note that $q$ also depends on $S$ (since it is a function of $m(S)$), but this does not effect our analysis, and hence this dependence is not explicit.
As in the analysis of the case that $\onk\in\krange$, since for each $k$-clique assigned to a vertex $u\in S$, the number of corresponding $k$-tuples in $\mC(S)$ is $(k-2)!$, we have that if $T$ is degrees-typical, then
\[\Pr[\chi_i(S)=1] \leq (1+\oeps) \cdot \frac{\alpha_{\Pip}(S)\cdot (k-2)!}{m(S)\cdot (\thres)^{k-2}}\;.
\]
By the definition of the $Y_i(S)$ variables and Equation~\eqref{eqn:exp_hnk}, it follows that
\[ \EX\left[\hnk \right] \leq  \EX_{S}\left[(1+\oeps)\cdot \frac{\alpha_{\Pip}(S) }{s/n} \right]
= (1+\oeps)\cdot\frac{n}{s} \cdot \EX_{S}\left[\alpha_{\Pip}(S)\right] \;.\]

By Claim~\ref{clm:S-typical}, for any partition $P$, $\EX[\ap(S)]\leq \frac{s}{n}\cdot \clk$. Therefore, if $T$ is degrees-typical,  then $\EX\left[\hnk \right]\leq (1+\oeps)\clk \;.$
Finally, by Markov's inequality (and recalling that $\oeps = \eps/5$),
\[  
\Pr\left[\hnk > (1+\eps/2)(1+\oeps)\clk\right]< 1-\eps/2\;. \]
As noted  in Subsection~\ref{subsec:main-alg}, we can assume that $\oeps \geq 1/\om^{k/2}$.
Since $T$ is not degrees-typical with probability at most $\min\{1/\om^{k/2}, \odelta \}$,
by taking a union bound and by the setting of $\oeps=\eps/5$, it holds that with probability at least $\eps/2-\setgamma >\eps/4$, the algorithm returns a value $\hnk$ such that $\hnk\leq (1+\eps/2)(1+\oeps)\clk\leq(1+\eps)\clk$.

\subsubsection{The expected query complexity and running time}
\sloppy
By Lemma~\ref{lem:T-typical} and the assumption that $\onk \leq \om^{k/2}$, the invocation of the procedure \link{Sample-degrees-typical} takes $O\left(
\frac{n}{\onk^{1/k}}\cdot \frac{k^2\cdot \log(n/\gamma)\cdot\log(1/\gamma)}{\oeps^2 }\right)$
 time, for $\gamma=\Theta\left(\min\{1/\om^{k/2}, 1/\delta\}\right)$.
The sampling of  $S$ and the computation of $m(S)$ and $D(S)$ take time
$O\left(s\right) =O\left(\frac{k \cdot n}{\onk^{1/k}} \cdot \frac{\ln(1/\odelta)}{\oeps^{\hspace{1pt}2+1/k}} \right)$.
By Lemma~\ref{lem:sample-clique}, each invocation of \link{Sample-a-clique} in Step~\ref{step:sample-clique} takes time $O(k^2)$, and  Step~\ref{step:chi_i} takes constant time. Therefore, excluding the invocations of the \textsf{Is-\popular} procedure,
the running time of the for loop in Step~\ref{step:q_loop} is $O\left(k^2 \cdot q\right)$.
Since $\EX[m(S)]=\frac{s}{n}\cdot m $ and by the setting of $q$, it follows that the expected running time of the for loop is $O\left(\frac{ m \cdot {\om}^{(k-2)/2}}{\onk} \cdot
\frac{k^2 \cdot \ln(1/\odelta)}{(k-2)! \cdot \oeps^{\hspace{1pt}2}} \right)$.
It remains to bound the running time resulting from the invocations of the procedure \textsf{Is-\popular}.

We first bound the expected number of invocations of \link{Is-\popular}  when the multiset $T$ is degrees-typical.
Let $I$ denote the number of invocations of  \link{Is-\popular}.
Similarly to the analysis of the case $\onk> \clk$, let $z_i(S)$
be a $0/1$ random variable that is defined as follows: $z_i(S)=1$ if (and only if) the procedure \link{Sample-a-clique} returned a clique in the $i\th$ step of the for loop in Step~\ref{step:sample-clique} of the algorithm, conditioned on the set $S$.
As before, let \textsf{S-c} denote the random coins of  \link{Sample-a-clique}.
Here it is actually relevant that $q$ depends on $S$, and therefore we use the notation $q(S)$.
By the definition of $I$,
\ifnum\stoc=1
\begin{align*}
\EX[I] &=  \EX_{S}\left[  \EX_{\emph{S-c}}\left[ \sum_{i=1}^{q(S)} z_i(S) \right] \right] \\
&= \EX_{S}\left[  q(S) \cdot \EX_{\emph{S-c}}\left[  z_1(S) \right] \right] \\ &
= \EX_{S}\left[  \EX_{\emph{S-c}}\left[  q(S) \cdot z_1(S) \right] \right] \;.
\end{align*}
\else
\[\EX[I] =  \EX_{S}\left[  \EX_{\emph{S-c}}\left[ \sum_{i=1}^{q(S)} z_i(S) \right] \right] = \EX_{S}\left[  q(S) \cdot \EX_{\emph{S-c}}\left[  z_1(S) \right] \right] = \EX_{S}\left[  \EX_{\emph{S-c}}\left[  q(S) \cdot z_1(S) \right] \right] \;.
\]
\fi
Recall that if $T$ is degrees-typical, then by Lemma~\ref{lem:sample-clique}, an invocation of  \link{Sample-a-clique}  returns each $k$-tuple in $\mC(S)$ with probability in $(1\pm\oeps)\cdot \frac{1}{m(S)\cdot (\thres)^{k-2}}$.
Hence,
\[\Pr_{\;\textsf{S-c}}[z_i(S)=1] \leq (1+\oeps) \frac{|\mC(S)|}{m(S)\cdot (\thres)^{k-2}}\;.
\]
By the setting of $q(S)$,
\[  \EX_{\emph{S-c}}[q(S) \cdot z_1(S) ] \leq \frac{(1+\oeps)|\mC(S)|}{(1-\oeps)^3 \cdot (k-2)!  \cdot \onk \cdot \frac{s}{n} }\cdot \frac{10\ln(1/\odelta)}{\oeps^{\hspace{1pt}2}} \;.
\]
Since $\EX[|\mC(S)|] = \frac{s}{n}\cdot (k-2)!\cdot k \cdot \clk,$ it follows that
$\EX[I]\leq \frac{\clk}{\onk} \cdot \frac{50k\ln(1/\odelta)}{\oeps^{\hspace{1pt}2}}  \;.$
Therefore, by Lemma~\ref{lem:is-popular}, if $T$ is degrees-typical, then the expected running time resulting from all of the invocations of the \link{Is-\popular} procedure is
$O\left(\frac{\om^{k/2}}{\onk} \cdot  \frac{\clk}{\onk}\cdot\frac{k^2 \cdot 2^k\cdot \log(n/\odelta)\log(1/\odelta)}{(k-2)!\cdot\oeps^{2-1/k}}\right).$

If $T$ is not degrees-typical then we can no longer upper bound the success probability of  \link{Sample-a-clique}, implying that the number of invocations of \link{Is-\popular} can be $\Theta(q)$. However, since $T$ is not degrees typical with probability at most $\frac{1}{\om^{k/2}}$, this does not affect the expected running time resulting from the invocations of  \link{Is-\popular}.
The remaining steps take constant time, and therefore the total running time of the algorithm is
\sloppy
\ifnum\stoc=1
\begin{align*}
& O\Bigg(\frac{n}{\onk^{1/k}}\cdot \frac{k^2\cdot \log(n/\gamma)\cdot\log(1/\gamma)}{\eps^{2+1/k} }  \\ &  \;\;\;\; + \;\;\frac{\min\{m,\om\}\cdot \om^{(k-2)/2}}{\onk} \cdot \frac{\clk}{\onk}\cdot\frac{k^2\cdot 2^k\cdot \log(n/\odelta)\log(1/\odelta)}{(k-2)!\cdot\oeps^{\hspace{1pt}2-1/k}} \Bigg)
\end{align*}
\else
\[O\Bigg(\frac{n}{\onk^{1/k}}\cdot \frac{k^2\cdot \log(n/\gamma)\cdot\log(1/\gamma)}{\eps^{2+1/k} }   + \frac{\min\{m,\om\}\cdot \om^{(k-2)/2}}{\onk} \cdot \frac{\clk}{\onk}\cdot\frac{k^2\cdot 2^k\cdot \log(n/\odelta)\log(1/\odelta)}{(k-2)!\cdot\oeps^{\hspace{1pt}2-1/k}} \Bigg)\]
\fi

which is
\[
O\left(\frac{n}{\onk^{1/k}}+ \frac{\max\{m,\om\}\cdot \om^{(k-2)/2}}{\onk} \cdot \frac{\clk}{\onk}\right)
\cdot\poly(\log n,1/\eps,\log(1/\delta),k), \]
since $\gamma=\Theta\left(\min\{1/\om^k, \delta\}\right)$.

Finally, as discussed in the beginning of Subsection~\ref{subsec:main-alg},
	if $q$ is greater than $\om$ then the algorithm may query beforehand for all the edges incident to $S$ and their neighbors, and if it views more than $\om$ edges then it can abort. Hence, the number of queries is always bounded by $O\left(\frac{n}{\onk^{1/k}}\right)\cdot \poly(\log n, 1/\eps,k)+\min\{m,\om\}$	\;.

\fi 

%% file: search.tex
	\subsection{The search algorithm} \label{subsec:search}
	In this subsection we describe an algorithm that returns an estimate of the number of $k$-cliques in a graph $G$ without prior knowledge on the number of edges or $k$-cliques in the graph.
We prove this by establishing a more general claim:

\begin{theorem} \label{thm:search}
	Let  $\invokeA$ be an algorithm that is given parameters $\oa,\eps, \delta$ and possibly an additional set of parameters denoted $\overrightarrow V$,
	for which the following holds.
\begin{enumerate}
	\item\label{it:oa-good} If $\oa \in [\va/4,\va]$, then with probability at least $1-\delta$, $\mA$ returns a value $\ha$ such that
	$\ha\in(1\pm \eps)\va$.
	\item\label{it:oa-big} If $\oa > \va$, then $\mA$ returns a value $\ha$, such that with probability at least $\eps/4$, $\ha \leq (1+\eps)\va$.
	\item\label{it:exp-t-A} The expected running time of $\mA$, denoted $\ert{\invokeA}$,
 is monotonically non-increasing with $\oa$ and furthermore, if $\oa < \va$, then
$\ert{\invokeA} \leq \ert{\mA\left(\va,\eps,\delta, \overrightarrow{V}\right)}\cdot (\va/\oa)^\ell$
for some constant $\ell > 0$.
	\item\label{it:max-t-A} The maximal running time of $\mA$ is $D$.
\end{enumerate}
Then there exists an algorithm $\mA'$ that, when given an upper bound $B$ on $\va$, a parameter $\eps$ and a set of parameters $\overrightarrow{V}$, returns a value $X$ such that the following holds.
\begin{enumerate}
	\item\label{it:Ap-good}  $\mA'(B,\eps,\overrightarrow{V})$  returns a value $X$ such that $X \in (1\pm\eps)\va$ with probability at least $4/5$.
	\item\label{it:exp-t-Ap} The expected running time of $\mA'\left(B,\eps,\overrightarrow{V}\right)$ is $\ert{\mA\left(\va,\eps,\delta, \overrightarrow{V}\right)} \cdot \poly( \log B, 1/\eps, \ell)$ for $\delta =\Theta\left(\frac{\eps}{2^{\ell}(\ell + \log \log (B))}\right)\;.$
	\item\label{it:max-t-Ap} The maximal running time of $\mA'$ is $D \cdot \poly( \log B, 1/\eps, \ell)$.
\end{enumerate}
\end{theorem}

The algorithm $\mA'$ referred to by the theorem is provided in Figure~\ref{fig:search-alg}. We note that the algorithm and the proof of \dcom{Theorem~\ref{thm:search}} are a  direct  generalization of the search algorithm \textsf{Estimate} and the proof of Theorem 12 in~\cite{ELRS}. However, this generalization may be useful as a ``black box'' in future work. We assume that $\eps \leq 1/4$, since otherwise we can simply set $\eps$ to $1/4$.

\confFigure{
	{\bf $\bm{\mA'}(\mA,B, \eps, \overrightarrow{V})$} 
			\smallskip
			\begin{compactenum}	
				\item Let $\ta=B$.
				\item Let $\delta'={1}/{(5 \cdot  2^\ell)}.$
				\item While $\ta \geq 1$ do: \label{step:while}
				\begin{compactenum}
					\item For $\oa=B, B/2, \ldots, \ta$ do: \label{step:search-for}
					\begin{compactenum}
						\item Let $r= (4/\eps)\cdot \ln(2\log^2(B)/\delta')$ and let $\delta=\delta'/(2r)$. \label{step:for}
						\item Let $X_{\oa}$ be the minimum value returned over $r$ invocations of $\invokeAp$. \label{step:min}
						\item If $X_{\oa} \geq (1+\eps)\oa$ then \textbf{return} $\oa$.
					\end{compactenum}
					\item Let $\ta=\ta/2$.
				\end{compactenum}
				\item \textbf{Return} \emph{fail}.
			\end{compactenum}
	}
{The search algorithm. \label{fig:search-alg}}

The following definition will be useful in the proof of the theorem.

\begin{definition} \label{def:good-estimate}
We say that a value $X$ is a \textsf{good estimate} of $\va$ if $X\in(1\pm \eps)\va$.
\end{definition}

\begin{proofof}{Theorem~\ref{thm:search}}
Our search algorithm has two nested loops running with decreasing values of ``guesses'' for the value of $\va$. The outer loop runs with $\ta$, which is our current guess for the value of $\va$, and
the purpose of the inner for loop is to enhance the success probability of the algorithm when $\ta$ ''passes'' the good guess of $\va$ and runs with values smaller than $\va/4$. We provide the full details subsequently,
 and start by considering only the outer while loop. Namely, imagine for now that  instead of the for loop in Step~\ref{step:search-for}, we have the command ``Let $\oa=\ta$'' and the rest of the algorithm is as described
 in Figure~\ref{fig:search-alg}.

First consider iterations of the while loop  for which $\ta > \va $. By Item~\ref{it:oa-big} in the properties of Algorithm $\mA$, for values $\ta$ such that $\ta>\va$, the probability that $\invokeA$ returns a value such that
 $\ha>(1+\eps)\va$ is
at most $1-\eps/4$. Hence, the probability that the minimum value returned over
 $r$ invocations, that is, $X_{\ta}$, satisfies $X_{\ta} > (1+\eps)\va$, is at most $(1-\eps/4)^r = (1-\eps/4)^{(4/\eps)\cdot \ln(2\log^2(B)/\delta')}< \frac{\delta'}{2\log^2(B)}$. It follows that
 for each value $\ta>\va$,
 with probability at least $1-\frac{\delta'}{2\log^2(B)}$, 
\[ X_{\ta}<(1+\eps) \va < (1+\eps)\ta. \]
This implies that for each value $\ta>\va$, with probability at least $1-\frac{\delta'}{2\log^2(B)}$, the algorithm $\mA'$ will 
continue to run with $\ta=\ta/2$.

Now consider values of $\ta$ such that $\ta\in[\va/4,\va]$.
By Item~\ref{it:oa-good} in the properties of Algorithm $\mA$,  if $\ta\in[\va/4,\va]$, then with probability at least $1-\delta'$, $\invokeA$ returns a value $\ha$, such that  $\ha \in (1\pm\eps)\va$. By the setting $\delta = \delta'/(2r)$, and by taking a union bound, with probability at least $1-\delta'/2$, all the $r$ invocations return a good estimate of $\va$, implying that
 $X_{\ta} \in (1\pm \eps)\va$ \dcom{(i.e., $X_{\ta}$ is a good estimate of $\va$}.
 \dcom{Note that in such a case,}
\[X_{\ta}\geq (1-\eps)\va \geq (1-\eps)\cdot 2\ta \geq (1+\eps)\ta . \]
(Here we have used the assumption that $\eps\leq 1/4$.)
Hence, once $\mA'$ reaches a value $\ta\in[\va/4,\va/2]$, with probability at least $1-\delta'/2$,
it 
returns a value $X_{\ta}$ that is a good estimate of $\va$.

Finally, if $\mA'$ reaches values $\oa<  \va/4$, we no longer have a guarantee on the probability that $\mA$ returns a value $\ha$ or on the quality of the estimate $\ha$. Hence, we also have an inner for loop so that whenever we halve the guess $\ta$ we first run with all the values $B, \ldots,\ta.$ This ensures that even if the algorithm did not return a value when running with a good guess $\ta \in [\va/4, \va/2]$, we can still bound the probability that it will continue to run with decreasing values of $\ta$.

We now consider the original algorithm with both the outer while loop and the inner for loop as described in Figure~\ref{fig:search-alg}.
There are at most $\log(B)$ invocations of the while loop with values  $\ta>\va$, implying that
there are at most $\log^2(B)$ invocations of the  for loop in Step~\ref{step:for} with a value $\oa>\va$. Therefore, by the above analysis, the probability that $\mA'$  returns a value that is not a good estimate of $\va$ in these invocations is at most $\delta'/2$. Hence, with probability at least $1-\delta'/2$ we will reach an invocation in which $\ta\in[\va/4,\va/2]$ and $\oa\in[\va/4,\va/2]$, for which with probability at least $1-\delta'/2$ the algorithm $\mA$ returns a value $X_{\oa}$ that  is a good estimate of $\va$. By taking a union bound and by the setting of $\delta'$, the algorithm $\mA'$ returns a value that is a good estimate of $\va$ with probability at least $1-\delta'>4/5$.

We turn to analyze the running time of $\mA'$. By Item~\ref{it:exp-t-A},
for values of $\oa$ such that $\oa\geq \va$, the expected running time of $\invokeA$ is $\ert{\mA\left(\va,\eps,\delta, \overrightarrow{V}\right)}$.
By the above analysis,
once $\mA'$ reaches a value $\ta< \va/4$, since we run with all values $\oa=B, \ldots, \ta$, 
the probability that the algorithm will halve the guess to $\ta=\ta/2$ is at most $\delta'$. Hence, the probability of invoking $\mA$ with a value $\oa = \va/2^z$ is at most $(\delta')^z$.
By Item~\ref{it:exp-t-A} in the properties of Algorithm $\mA$, when running with values $\oa<\va$, it holds that
$\ert{\invokeA} \leq \ert{\mA\left(\va,\eps,\delta, \overrightarrow{V}\right)}\cdot (\va/\oa)^\ell$.
Therefore, the expected running time of $\mA'$ is upper bounded by
\ifnum\stoc=1
\begin{align*}
\log^2(B) &\cdot r \cdot \ert{\mA\left(\va,\eps,\delta, \overrightarrow{V}\right)}
\\ \;\;\;\;& +
 \sum_{z=1}^{\log B}(\delta')^z \cdot  2^{\ell \cdot z} \cdot \ert{\mA\left(\va,\eps,\delta, \overrightarrow{V}\right)}  \\
&
\;\;\;\;\;\;\;
= \ert{\mA\left(\va,\eps,\delta, \overrightarrow{V}\right)}   \cdot \poly(\log B, 1/\eps, \ell)\;,
\end{align*}
\else
\begin{align*}
\log^2(B) &\cdot r \cdot \ert{\mA\left(\va,\eps,\delta, \overrightarrow{V}\right)} +
\sum_{z=1}^{\log B}(\delta')^z \cdot  2^{\ell \cdot z} \cdot \ert{\mA\left(\va,\eps,\delta, \overrightarrow{V}\right)} 
\\ &= \ert{\mA\left(\va,\eps,\delta, \overrightarrow{V}\right)}   \cdot \poly(\log B, 1/\eps, \ell)\;,
\end{align*}
\fi
where the first term is due to the invocations in which $\oa\geq \va$ and the second is due to invocations in which $\oa<\va$, and the equality is due to the setting of $\delta'< 1/2^{\ell}$.
Therefore, Item~\ref{it:exp-t-Ap} of the theorem holds. The proof of Item~\ref{it:max-t-Ap} is immediate.
\end{proofof}

\medskip

The following is a corollary of Theorem~\ref{thm:correct} and Theorem~\ref{thm:search}
and is a restatement of Theorem~\ref{thm:main}.
	\begin{corollary}\label{clm:markov} \sloppy
There exists an algorithm that, given $n$, $k\geq 3$, and query access to a graph $G$, returns a value
		$X$ such that $X \in (1\pm\eps)\clk$ with probability at least $2/3$. The expected
		running time of the algorithm is  $O\left({\frac{n}{\clk^{1/k}} + \frac{m^{k/2}}{\clk \cdot (k-2)!}  }\right)\cdot \poly(\log n,1/\eps,k) $,
		and its expected query complexity is $O\left({\frac{n}{\clk^{1/k}} + \min\left\{\frac{m^{k/2}}{\clk \cdot (k-2)!}\dcom{,m} \right\} }\right)\cdot \poly(\log n,1/\eps,k) \;.$
		
	\end{corollary}
	\begin{proof}
\sloppy
We start by obtaining an estimate $\om$ of $m$ such that with probability at least $1-\min\{1/n,1/m^{k/2}\}$ it holds that $\om\in[(1-\eps/4)m,m]$. This can be done by invoking the algorithm of Goldreich and Ron \cite{GR08} for estimating the number of edges $\Theta(\log(n^{2k}))$ times and taking the median value returned. Next we invoke $\mA'(\textsf{Approximate-cliques}, B=\min\{n^k, \om^{k/2}\} , \eps,  \overrightarrow{V})$ with $\overrightarrow{V}=(n,k,\om)$ and return the value $X$ returned by $\mA'$.
		
		By Theorem~\ref{thm:correct},  if $\om\in[(1-\eps/5)m, m]$, then 
\link{Approximate-cliques} satisfies the conditions required from Algorithm $\mA$ in Theorem~\ref{thm:search}, with  $\ell=2$ and $D=n+m^{k/2}$ (since this is the maximal running time when $\onk=1$). Hence, by Theorem~\ref{thm:search} and by the union bound,  $\mA'$ returns a value $X$ such that with probability at least $4/5-1/m^k>2/3$, it holds that $X \in (1\pm\eps)\clk$.
		
		By~\cite{GR08}, the first step of approximating the number of edges $\om$ with success probability at least $1-\min\{1/n,1/m^k \}$ takes time $O\left(\frac{n}{\sqrt m}\right)\cdot \poly(\log n,1/\eps,k)$ which by Claim~\ref{clm:u.b. nk} is at most $O\left(\frac{n}{\clk^{1/k}}\right)\cdot \poly(\log n,1/\eps,k)$.
		When $\om \in [1-(\eps/5)m, m]$, by Item~\ref{item:thm-running} in Theorem~\ref{thm:correct}, and by Item~\ref{it:exp-t-Ap} in the properties of Algorithm $\mA'$,	the expected running time 
f the algorithm is
		\[O\left({\frac{n}{\clk^{1/k}}+\frac{m^{k/2}}{\clk} }\right) \cdot \poly\left(\log n, 1/\eps,k\right) \;.\]
		Since  $\om \notin [(1-\eps/5)m,m]$ with probability at most $\min\{1/n,1/m^{k/2}\}$ and
		the maximal running time of \link{Approximate-cliques} is  at most $n+m^{k/2}$ \tcom{up to $\poly(\log n,1/\eps,k)$ factors}, this event does not effect the expected query complexity.
		
		By Item~\ref{item:thm-running} in Theorem~\ref{thm:correct},
		Algorithm \link{Approximate-cliques} never performs more than
		$O\left(\frac{n}{\onk^{1/k}}\right)\cdot\poly(\log n,1/\eps,\log(1/\delta),k) + \min\{m,\om\}$  queries. It follows that the expected query complexity is
\[O\left( \frac{n}{\clk^{1/k}}+\min\left\{\frac{m^{k/2}}{\clk} \dcom{,m}  \right\}  \right) \cdot \poly\left(\log n, 1/\eps,k\right) \;,\]
		as claimed.
	\end{proof}

%% file: k_cliques.bbl
\newcommand{\etalchar}[1]{$^{#1}$}
\begin{thebibliography}{MTW{\etalchar{+}}04}

\bibitem[ABW15]{ABW15}
A.~Abboud, A.~Backurs, and V.~V. Williams.
\newblock If the current clique algorithms are optimal, so is valiant's parser.
\newblock In {\em Proceedings of the Symposium on Foundations of Computer
  Science (FOCS)}, pages 98--117, 2015.

\bibitem[AGM12]{AhGuMc12}
K.~J. Ahn, S.~Guha, and A.~McGregor.
\newblock Graph sketches: sparsification, spanners, and subgraphs.
\newblock In {\em Proceedings of the Symposium on Principles of Database
  Systems (PODS)}, pages 5--14, 2012.

\bibitem[Avr10]{avron2010counting}
H.~Avron.
\newblock Counting triangles in large graphs using randomized matrix trace
  estimation.
\newblock In {\em Workshop on Large-scale Data Mining: Theory and Applications
  (LDMTA)}, volume~10, pages 10--9, 2010.

\bibitem[BBCG08]{becchetti2008efficient}
L.~Becchetti, P.~Boldi, C.~Castillo, and A.~Gionis.
\newblock Efficient semi-streaming algorithms for local triangle counting in
  massive graphs.
\newblock In {\em Proceedings of the International Conference on Knowledge
  Discovery and Data Mining (SIGKDD)}, pages 16--24, 2008.

\bibitem[BHLP11]{BerryHLP11}
J.~W. {Berry}, B.~{Hendrickson}, R.~A. {LaViolette}, and C.~A. {Phillips}.
\newblock {Tolerating the Community Detection Resolution Limit with Edge
  Weighting}.
\newblock {\em Physical Review E}, 83(5):056119, May 2011.

\bibitem[Bur04]{Burt04}
R.~S. Burt.
\newblock Structural holes and good ideas.
\newblock {\em American Journal of Sociology}, 110(2):349--399, 2004.
\newblock URL: \url{http://www.jstor.org/stable/10.1086/421787}.

\bibitem[CC11]{chu2011triangle}
S.~Chu and J.~Cheng.
\newblock Triangle listing in massive networks and its applications.
\newblock In {\em Proceedings of the International Conference on Knowledge
  Discovery and Data Mining (SIGKDD)}, pages 672--680, 2011.

\bibitem[CDK{\etalchar{+}}16]{ChDa+16}
F.~Chierichetti, A.~Dasgupta, R.~Kumar, S.~Lattanzi, and T.~Sarlos.
\newblock On sampling nodes in a network.
\newblock In {\em Conference on the World Wide Web (WWW)}, pages 471--481,
  2016.

\bibitem[CEF{\etalchar{+}}05]{DBLP:journals/siamcomp/CzumajEFMNRS05}
A.~Czumaj, F.~Erg{\"{u}}n, L.~Fortnow, A.~Magen, I.~Newman, R.~Rubinfeld, and
  C.~Sohler.
\newblock Approximating the weight of the {E}uclidean minimum spanning tree in
  sublinear time.
\newblock {\em SIAM Journal on Computing}, 35(1):91--109, 2005.

\bibitem[Che52]{chernoff}
H.~Chernoff.
\newblock A measure of asymptotic efficiency for tests of a hypothesis based on
  the sum of observations.
\newblock {\em The Annals of Mathematical Statistics}, pages 493--507, 1952.

\bibitem[CN85]{ChNi85}
N.~Chiba and T.~Nishizeki.
\newblock Arboricity and subgraph listing algorithms.
\newblock {\em SIAM Journal on Computing}, 14(1):210--223, 1985.

\bibitem[Col88]{Co88}
J.~S. Coleman.
\newblock Social capital in the creation of human capital.
\newblock {\em American Journal of Sociology}, 94:S95--S120, 1988.
\newblock URL: \url{http://www.jstor.org/stable/2780243}.

\bibitem[CRT05]{DBLP:journals/siamcomp/ChazelleRT05}
B.~Chazelle, R.~Rubinfeld, and L.~Trevisan.
\newblock Approximating the minimum spanning tree weight in sublinear time.
\newblock {\em SIAM Journal on Computing}, 34(6):1370--1379, 2005.

\bibitem[CS09]{DBLP:journals/siamcomp/CzumajS09}
A.~Czumaj and C.~Sohler.
\newblock Estimating the weight of metric minimum spanning trees in sublinear
  time.
\newblock {\em SIAM Journal on Computing}, 39(3):904--922, 2009.

\bibitem[DKS14]{DKS}
A.~Dasgupta, R.~Kumar, and T.~Sarlos.
\newblock On estimating the average degree.
\newblock In {\em Conference on the World Wide Web (WWW)}, pages 795--806. ACM,
  2014.

\bibitem[EG04]{EG04}
F.~Eisenbrand and F.~Grandoni.
\newblock On the complexity of fixed parameter clique and dominating set.
\newblock {\em Theoretical Computer Science}, 326(1-3):57--67, 2004.

\bibitem[ELRS15]{ELRS}
T.~Eden, A.~Levi, D.~Ron, and C~Seshadhri.
\newblock Approximately counting triangles in sublinear time.
\newblock In {\em Proceedings of the Symposium on Foundations of Computer
  Science (FOCS)}, pages 614--633, 2015.

\bibitem[ELS13]{ELS13}
D.~Eppstein, M.~L{\"{o}}ffler, and D.~Strash.
\newblock Listing all maximal cliques in large sparse real-world graphs.
\newblock {\em {ACM} Journal of Experimental Algorithmics}, 18:3--1, 2013.
\newblock URL: \url{http://doi.acm.org/10.1145/2543629}, \href
  {http://dx.doi.org/10.1145/2543629} {\path{doi:10.1145/2543629}}.

\bibitem[EM02]{EcMo02}
J.~P. Eckmann and E.~Moses.
\newblock Curvature of co-links uncovers hidden thematic layers in the {World
  Wide Web}.
\newblock {\em Proceedings of the National Academy of Sciences},
  99(9):5825--5829, 2002.
\newblock \href {http://dx.doi.org/10.1073/pnas.032093399}
  {\path{doi:10.1073/pnas.032093399}}.

\bibitem[ER17]{ER_CC}
T.~Eden and W.~Rosenbaum.
\newblock Lower bounds for approximating graph parameters via communication
  complexity.
\newblock {\em CoRR}, abs/1709.04262, 2017.
\newblock URL: \url{http://arxiv.org/abs/1709.04262}, \href
  {http://arxiv.org/abs/1709.04262} {\path{arXiv:1709.04262}}.

\bibitem[ER18]{EdenRosenbaum}
T.~Eden and W.~Rosenbaum.
\newblock On sampling edges almost uniformly.
\newblock In {\em 1st Symposium on Simplicity in Algorithms, {SOSA} 2018,
  January 7-10, 2018, New Orleans, LA, {USA}}, pages 7:1--7:9, 2018.
\newblock URL: \url{https://doi.org/10.4230/OASIcs.SOSA.2018.7}, \href
  {http://dx.doi.org/10.4230/OASIcs.SOSA.2018.7}
  {\path{doi:10.4230/OASIcs.SOSA.2018.7}}.

\bibitem[ERS17a]{ERS_cliques}
T.~Eden, D.~Ron, and C.~Seshadhri.
\newblock On approximating the number of $k$-cliques in sublinear time.
\newblock {\em CoRR}, abs/1707.04858v1, 2017.
\newblock URL: \url{http://arxiv.org/abs/1707.04858v1}, \href
  {http://arxiv.org/abs/1707.04858v1} {\path{arXiv:1707.04858v1}}.

\bibitem[ERS17b]{ERS16}
Talya Eden, Dana Ron, and C.~Seshadhri.
\newblock Sublinear time estimation of degree distribution moments: The
  degeneracy connection.
\newblock In {\em 44th International Colloquium on Automata, Languages, and
  Programming, {ICALP} 2017, July 10-14, 2017, Warsaw, Poland}, pages
  7:1--7:13, 2017.
\newblock URL: \url{https://doi.org/10.4230/LIPIcs.ICALP.2017.7}, \href
  {http://dx.doi.org/10.4230/LIPIcs.ICALP.2017.7}
  {\path{doi:10.4230/LIPIcs.ICALP.2017.7}}.

\bibitem[Fei06]{feige2006sums}
U.~Feige.
\newblock On sums of independent random variables with unbounded variance and
  estimating the average degree in a graph.
\newblock {\em SIAM Journal on Computing}, 35(4):964--984, 2006.

\bibitem[FFF15]{FFF15}
I.~Finocchi, M.~Finocchi, and E.~G. Fusco.
\newblock Clique counting in mapreduce: Algorithms and experiments.
\newblock {\em ACM Journal of Experimental Algorithmics}, 20:1--7, 2015.
\newblock URL: \url{http://doi.acm.org/10.1145/2794080}, \href
  {http://dx.doi.org/10.1145/2794080} {\path{doi:10.1145/2794080}}.

\bibitem[FVC10]{foucault2010friend}
B.~{Foucault Welles}, A.~{Van Devender}, and N.~Contractor.
\newblock Is a friend a friend?: Investigating the structure of friendship
  networks in virtual worlds.
\newblock In {\em CHI Extended Abstracts on Human Factors in Computing
  Systems}, pages 4027--4032, 2010.

\bibitem[Gol17]{G17-book}
Oded Goldreich.
\newblock {\em Introduction to Property Testing}.
\newblock Cambridge University Press, 2017.

\bibitem[GR08]{GR08}
O.~Goldreich and D.~Ron.
\newblock Approximating average parameters of graphs.
\newblock {\em Random Structures and Algorithms}, 32(4):473--493, 2008.

\bibitem[GRS11]{GRS11}
M.~Gonen, D.~Ron, and Y.~Shavitt.
\newblock Counting stars and other small subgraphs in sublinear-time.
\newblock {\em SIAM Journal on Discrete Mathematics}, 25(3):1365--1411, 2011.

\bibitem[HKNO09]{hassidim2009local}
A.~Hassidim, J.~A. Kelner, H.~N. Nguyen, and K.~Onak.
\newblock Local graph partitions for approximation and testing.
\newblock In {\em Proceedings of the Symposium on Foundations of Computer
  Science (FOCS)}, pages 22--31, 2009.

\bibitem[HL70]{HoLe70}
P.~W. Holland and S.~Leinhardt.
\newblock A method for detecting structure in sociometric data.
\newblock {\em American Journal of Sociology}, 76:492--513, 1970.

\bibitem[JG05]{jowhari2005new}
H.~Jowhari and M.~Ghodsi.
\newblock New streaming algorithms for counting triangles in graphs.
\newblock In {\em Proceedings of the International Conference Computing and
  Combinatorics (COCOON)}, pages 710--716. Springer, 2005.

\bibitem[JRBT12]{JRT12}
M.~O. Jackson, T.~Rodriguez-Barraquer, and X.~Tan.
\newblock Social capital and social quilts: Network patterns of favor exchange.
\newblock {\em American Economic Review}, 102(5):1857?1897, 2012.

\bibitem[JS17]{JS17}
S.~Jain and C.~Seshadhri.
\newblock A fast and provable method for estimating clique counts using
  tur\'{a}n's theorem.
\newblock In {\em Conference on the World Wide Web (WWW)}, pages 441--449,
  2017.

\bibitem[KMPT12]{kolountzakis2012efficient}
M.~N. Kolountzakis, G.~L. Miller, R.~Peng, and C.~E. Tsourakakis.
\newblock Efficient triangle counting in large graphs via degree-based vertex
  partitioning.
\newblock {\em Internet Mathematics}, 8(1-2):161--185, 2012.

\bibitem[KMSS12]{KaMeSaSu12}
D.~M. Kane, K.~Mehlhorn, T.~Sauerwald, and H.~Sun.
\newblock Counting arbitrary subgraphs in data streams.
\newblock In {\em International Colloquium on Automata, Languages, and
  Programming (ICALP)}, pages 598--609, 2012.

\bibitem[MR09]{DBLP:journals/talg/MarkoR09}
S.~Marko and D.~Ron.
\newblock Approximating the distance to properties in bounded-degree and
  general sparse graphs.
\newblock {\em ACM Transactions on Algorithms}, 5(2):22, 2009.

\bibitem[MSOI{\etalchar{+}}02]{milo2002network}
R.~Milo, S.~Shen-Orr, S.~Itzkovitz, N.~Kashtan, D.~Chklovskii, and U.~Alon.
\newblock Network motifs: simple building blocks of complex networks.
\newblock {\em Science}, 298(5594):824--827, 2002.

\bibitem[MTW{\etalchar{+}}04]{marsaglia2004fast}
G.~Marsaglia, W.~W. Tsang, J.~Wang, et~al.
\newblock Fast generation of discrete random variables.
\newblock {\em Journal of Statistical Software}, 11(3):1--11, 2004.

\bibitem[NO08]{nguyen2008constant}
H.~N. Nguyen and K.~Onak.
\newblock Constant-time approximation algorithms via local improvements.
\newblock In {\em Proceedings of the Symposium on Foundations of Computer
  Science (FOCS)}, pages 327--336, 2008.

\bibitem[NP85]{NP85}
J.~Ne{\v{s}}t{\v{r}}il and S.~Poljak.
\newblock On the complexity of the subgraph problem.
\newblock {\em Commentationes Mathematicae Universitatis Carolinae},
  26(2):415--419, 1985.

\bibitem[ORRR12]{onak2012near}
K.~Onak, D.~Ron, M.~Rosen, and R.~Rubinfeld.
\newblock A near-optimal sublinear-time algorithm for approximating the minimum
  vertex cover size.
\newblock In {\em Proceedings of the Symposium on Discrete Algorithms (SODA)},
  pages 1123--1131, 2012.

\bibitem[Por00]{portes2000social}
Alejandro Portes.
\newblock Social capital: Its origins and applications in modern sociology.
\newblock In Eric~L. Lesser, editor, {\em Knowledge and Social Capital}, pages
  43 -- 67. Butterworth-Heinemann, Boston, 2000.
\newblock URL:
  \url{http://www.sciencedirect.com/science/article/pii/B9780750672221500064},
  \href {http://dx.doi.org/https://doi.org/10.1016/B978-0-7506-7222-1.50006-4}
  {\path{doi:https://doi.org/10.1016/B978-0-7506-7222-1.50006-4}}.

\bibitem[PR07]{DBLP:journals/tcs/ParnasR07}
M.~Parnas and D.~Ron.
\newblock Approximating the minimum vertex cover in sublinear time and a
  connection to distributed algorithms.
\newblock {\em Theoretical Computer Science}, 381(1-3):183--196, 2007.

\bibitem[SKP12]{SeKoPi11}
C.~Seshadhri, T.~G. Kolda, and A.~Pinar.
\newblock Community structure and scale-free collections of {Erd\"os-R\'enyi}
  graphs.
\newblock {\em Physical Review E}, 85(5):056109, May 2012.
\newblock \href {http://dx.doi.org/10.1103/PhysRevE.85.056109}
  {\path{doi:10.1103/PhysRevE.85.056109}}.

\bibitem[SPK13]{SePiKo13}
C.~Seshadhri, A.~Pinar, and T.~G. Kolda.
\newblock Fast triangle counting through wedge sampling.
\newblock In {\em Proceedings of the International Conference on Data Mining
  (ICDM)}, volume~4, page~5, 2013.
\newblock URL: \url{http://arxiv.org/abs/1202.5230}.

\bibitem[SV11]{SuVa11}
S.~Suri and S.~Vassilvitskii.
\newblock Counting triangles and the curse of the last reducer.
\newblock In {\em Proceedings of the International Conference on World Wide Web
  (WWW)}, pages 607--614, 2011.
\newblock URL: \url{http://doi.acm.org/10.1145/1963405.1963491}, \href
  {http://dx.doi.org/10.1145/1963405.1963491}
  {\path{doi:10.1145/1963405.1963491}}.

\bibitem[SW05a]{ScWa05-2}
T.~Schank and D.~Wagner.
\newblock Approximating clustering coefficient and transitivity.
\newblock {\em Journal of Graph Algorithms and Applications}, 9:265--275, 2005.

\bibitem[SW05b]{ScWa05}
T.~Schank and D.~Wagner.
\newblock Finding, counting and listing all triangles in large graphs, an
  experimental study.
\newblock In {\em Experimental and Efficient Algorithms}, pages 606--609. 2005.

\bibitem[TKM11]{tsourakakis2011triangle}
C.~E. Tsourakakis, M.~N. Kolountzakis, and G.~L. Miller.
\newblock Triangle sparsifiers.
\newblock {\em Journal of Graph Algorithms and Applications}, 15(6):703--726,
  2011.

\bibitem[TKMF09]{tsourakakis2009doulion}
C.~E. Tsourakakis, U.~Kang, G.L. Miller, and C.~Faloutsos.
\newblock Doulion: counting triangles in massive graphs with a coin.
\newblock In {\em Proceedings of the International Conference on Knowledge
  Discovery and Data Mining (SIGKDD)}, pages 837--846, 2009.

\bibitem[TPT13]{TaPaTi13}
K.~Tangwongsan, A.~Pavan, and S.~Tirthapura.
\newblock Parallel triangle counting in massive streaming graphs.
\newblock In {\em Proceedings of the International Conference on Information
  and Knowledge Management (CIKM)}, pages 781--786. ACM, 2013.

\bibitem[Tso08]{tsourakakis2008fast}
C.~E. Tsourakakis.
\newblock Fast counting of triangles in large real networks without counting:
  Algorithms and laws.
\newblock In {\em International Conference on Data Mining (ICDM)}, pages
  608--617, 2008.

\bibitem[Tso15]{Ts15}
C.~E. Tsourakakis.
\newblock The k-clique densest subgraph problem.
\newblock In {\em Proceedings of the International Conference on World Wide Web
  (WWW)}, pages 1122--1132, 2015.
\newblock URL: \url{http://doi.acm.org/10.1145/2736277.2741098}, \href
  {http://dx.doi.org/10.1145/2736277.2741098}
  {\path{doi:10.1145/2736277.2741098}}.

\bibitem[Vas09]{V09}
V.~Vassilevska.
\newblock Efficient algorithms for clique problems.
\newblock {\em Information Processing Letters}, 109(4):254--257, 2009.

\bibitem[Wal74]{walker1974new}
A.~J. Walker.
\newblock New fast method for generating discrete random numbers with arbitrary
  frequency distributions.
\newblock {\em Electronics Letters}, 10(8):127--128, 1974.

\bibitem[Wal77]{walker1977efficient}
A.~J. Walker.
\newblock An efficient method for generating discrete random variables with
  general distributions.
\newblock {\em ACM Transactions on Mathematical Software}, 3(3):253--256, 1977.

\bibitem[YYI09]{yoshida2009improved}
Y.~Yoshida, M.~Yamamoto, and H.~Ito.
\newblock An improved constant-time approximation algorithm for maximum.
\newblock In {\em Proceedings of the Symposium on Theory of Computing (STOC)},
  pages 225--234, 2009.

\end{thebibliography}
